\documentclass[runningheads]{llncs}

\usepackage{ltl}
\usepackage{mathtools}

\usepackage{algorithm}
\usepackage[noend]{algpseudocode}
\usepackage{enumitem}
\usepackage{amssymb}
\usepackage{amsmath}
\usepackage{graphicx}
\usepackage{amsfonts} 
\usepackage{stmaryrd} 
\usepackage{xcolor}
\usepackage{booktabs}

\definecolor{linkC}{HTML}{710580}
\usepackage{hyperref}      
\hypersetup{
	colorlinks,
	citecolor=linkC,
	filecolor=red,
	linkcolor=linkC,
	urlcolor=blue
}             

\usepackage{tablefootnote}
\usepackage{makecell}
\usepackage{todonotes}
\usepackage{comment}
\usepackage{subcaption}
\usepackage{xcolor}
\usepackage{listings}
\usepackage{relsize}

\usepackage{thm-restate}

\usepackage{bussproofs}

\usepackage{booktabs}

\usepackage{pifont}

\usepackage{wrapfig}

\usepackage{multirow}

\usepackage{hyperref}     
\usepackage{cleveref}
\usepackage{orcidlink}

\usepackage{algorithm}
\allowdisplaybreaks
\newcommand{\niklas}[1]{{\footnotesize{\color{olive} [Niklas: #1]}}}

\newcommand{\traceassign}{\Pi}
\newcommand{\tracevars}{\mathcal{V}}
\newcommand{\sovars}{\mathfrak{V}}
\newcommand{\systemvar}{\mathfrak{S}}
\newcommand{\allvar}{\mathfrak{A}}
\renewcommand{\models}{\vDash}

\newcommand{\ap}{\text{AP}}

\newcommand{\U}{\LTLuntil}
\newcommand{\X}{\LTLnext}
\newcommand{\G}{\LTLglobally}
\newcommand{\F}{\LTLeventually}

\newcommand{\W}{\LTLweakuntil}

\newcommand{\largestSet}{\curlywedge}
\newcommand{\smallestSet}{\curlyvee}

\makeatletter
\newcommand{\superimpose}[2]{{%
		\ooalign{%
			\hfil$\m@th#1\@firstoftwo#2$\hfil\cr
			\hfil$\m@th#1\@secondoftwo#2$\hfil\cr
		}%
}}
\makeatother

\newcommand{\fptype}{\mathbin{\mathpalette\superimpose{{\curlywedge}{\curlyvee}}}}

\newcommand{\quant}{\mathbb{Q}}

\newcommand{\know}{\mathbb{K}}

\newcommand{\nat}{\mathbb{N}}

\newcommand{\traces}{\mathit{Traces}}

\newcommand{\set}[1]{\{#1\}}
\newcommand{\true}[0]{\mathit{true}}
\newcommand{\false}[0]{\mathit{false}}

\newcommand{\ldot}{\mathpunct{.}}

\newcommand{\aut}[1]{\ensuremath{\mathcal{#1}}}

\newcommand{\cmark}{\ding{51}}%
\newcommand{\xmark}{\ding{55}}%

\newcommand{\calT}{\mathcal{T}}

\newcommand{\calA}{\mathcal{A}}
\newcommand{\calB}{\mathcal{B}}
\newcommand{\calC}{\mathcal{C}}
\newcommand{\calD}{\mathcal{D}}
\newcommand{\calL}{\mathcal{L}}

\newcommand{\calK}{\mathsf{K}}
\newcommand{\sfC}{\mathsf{C}}

\newcommand{\modin}{\triangleright}

\newcommand{\sohyperltl}[0]{\text{\normalfont{Hyper\textsuperscript{2}LTL}}}%
\newcommand{\fphyperltl}[0]{\text{\normalfont{Hyper\textsuperscript{2}LTL\textsubscript{fp}}}} 
\newcommand{\oldsohyperltl}[0]{$\text{Hyper}^2\text{LTL}$}
\newcommand{\oldfphyperltl}[0]{$\text{Hyper}^2\text{LTL}_{fp}$}
\newcommand{\hyperltl}[0]{\text{\normalfont{HyperLTL}}}

\newcommand{\hyperqptl}[0]{\text{\normalfont{HyperQPTL}}}
\newcommand{\HyperQPTL}[0]{\text{\normalfont{HyperQPTL}}}

\makeatletter
\DeclareFontFamily{OMX}{MnSymbolE}{}
\DeclareSymbolFont{MnLargeSymbols}{OMX}{MnSymbolE}{m}{n}
\SetSymbolFont{MnLargeSymbols}{bold}{OMX}{MnSymbolE}{b}{n}
\DeclareFontShape{OMX}{MnSymbolE}{m}{n}{
    <-6>  MnSymbolE5
   <6-7>  MnSymbolE6
   <7-8>  MnSymbolE7
   <8-9>  MnSymbolE8
   <9-10> MnSymbolE9
  <10-12> MnSymbolE10
  <12->   MnSymbolE12
}{}
\DeclareFontShape{OMX}{MnSymbolE}{b}{n}{
    <-6>  MnSymbolE-Bold5
   <6-7>  MnSymbolE-Bold6
   <7-8>  MnSymbolE-Bold7
   <8-9>  MnSymbolE-Bold8
   <9-10> MnSymbolE-Bold9
  <10-12> MnSymbolE-Bold10
  <12->   MnSymbolE-Bold12
}{}

\let\llangle\@undefined
\let\rrangle\@undefined
\DeclareMathDelimiter{\llangle}{\mathopen}%
                     {MnLargeSymbols}{'164}{MnLargeSymbols}{'164}
\DeclareMathDelimiter{\rrangle}{\mathclose}%
                     {MnLargeSymbols}{'171}{MnLargeSymbols}{'171}
\makeatother

\definecolor{dkcyan}{rgb}{0.1, 0.3, 0.3}
\definecolor{dkgreen}{rgb}{0,0.3,0}
\definecolor{olive}{rgb}{0.5, 0.5, 0.0}
\definecolor{dkblue}{rgb}{0,0.1,0.5}

\definecolor{col:ln}{rgb}  {0.1, 0.1, 0.7}
\definecolor{col:str}{rgb} {0.8, 0.0, 0.0}
\definecolor{col:db}{rgb}  {0.9, 0.5, 0.0}
\definecolor{col:ours}{rgb}{0.0, 0.7, 0.0}

\definecolor{lightgreen}{RGB}{170, 255, 220}
\definecolor{darkbrown}{RGB}{121,37,0}

\colorlet{listing-comment}{gray}
\colorlet{operator-color}{darkbrown}
\colorlet{comment-color}{black!50}

\lstdefinelanguage{custom-lang}{
	keywords={let, in, where, match, with, when, if, then, else, for, repeat, return, to, do, from},
	keywordstyle=[1]\color{dkblue},
	morekeywords=[2]{verify, systemToNBA, LTLtoNBA, eProduct, uProduct},
	keywordstyle=[2]\color{dkgreen},
    morekeywords=[3]{underApprox, overApprox},
	keywordstyle=[3]\color{darkbrown},
	comment=[l][\color{comment-color}]{//},
	literate=%
	{=}{{{\color{operator-color}=}}}1
	{|}{{{\color{dkblue}|}}}1
	{:}{{{\color{dkblue}:}}}1
	{:=}{{{\color{dkblue}:=}}}1
    {@}{ }1
}

\lstdefinestyle{default}{
	escapeinside={(*}{*)},
	basicstyle=\ttfamily\fontsize{9.3}{10.3}\selectfont,
	columns=fullflexible,
	commentstyle=\sffamily\color{black!50!white},
	framexleftmargin=1em,
	framexrightmargin=1ex,
	keepspaces=true,
	keywordstyle=\color{dkblue},
	mathescape,
	numbers=left,
	numberblanklines=false,
	numbersep=1.25em,
	numberstyle=\relscale{0.8}\color{gray}\ttfamily,
	showstringspaces=true,
	stepnumber=1,
	xleftmargin=2em,
}

\lstnewenvironment{code}[1][]
{\small
	\lstset{
		style=default, 
		language=custom-lang,
		#1
	}
}
{}

\lstdefinelanguage{example-lang}{
	keywords={while,do},
	keywordstyle=[1]\bfseries,
	comment=[l][\color{comment-color}]{//},
	literate=%
	{<-}{{{\color{dkblue}$\leftarrow$}}}1
	{@}{ }1
}

\lstdefinestyle{example-style}{
	escapeinside={(*}{*)},
	basicstyle=\ttfamily\fontsize{8.4}{9.7}\selectfont,
	columns=fullflexible,
	commentstyle=\sffamily\color{black!50!white},
	framexleftmargin=0em,
	framexrightmargin=0ex,
	keepspaces=true,
	keywordstyle=\color{dkblue},
	mathescape,
	numbers=left,
	numberblanklines=false,
	showstringspaces=true,
	stepnumber=1,
	xleftmargin=0em,
	numbers=none
}

\lstnewenvironment{exampleCode}[1][]
{\small
	\lstset{
		style=example-style, 
		language=example-lang,
		#1
	}
}
{}

\usepackage{scalefnt}

\usetikzlibrary{arrows}
\usepackage{graphicx}                   
\usepackage{hyperref}                   
\usepackage[firstpage]{draftwatermark}  

\SetWatermarkAngle{0}
\SetWatermarkText{\raisebox{13.5cm}{%
  \hspace{0.1cm}%
  \href{https://zenodo.org/record/7878311}{\includegraphics{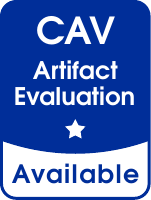}}%
  \hspace{9cm}%
  \includegraphics{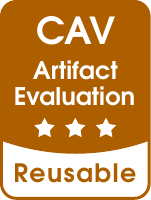}%
}}

\begin{document}

\title{Second-Order Hyperproperties}

\author{Raven Beutner\orcidlink{0000-0001-6234-5651}  \and Bernd Finkbeiner\orcidlink{0000-0002-4280-8441} \and Hadar Frenkel\orcidlink{0000-0002-3566-0338}  \and Niklas Metzger\orcidlink{0000-0003-3184-6335}}
\institute{CISPA Helmholtz Center for Information Security,\\ Saarbr\"ucken, Germany\\
\email{\{raven.beutner,finkbeiner,hadar.frenkel,\\niklas.metzger\} @cispa.de}}
\authorrunning{R. Beutner, B. Finkbeiner, H. Frenkel, N. Metzger}
\maketitle              

\begin{abstract}
We introduce \sohyperltl, a temporal logic for the specification of hyperproperties that allows for 
second-order quantification over sets of traces. Unlike first-order temporal logics for hyperproperties, such as HyperLTL,
\sohyperltl\ can express complex epistemic properties like common knowledge, Mazurkiewicz trace theory, and asynchronous hyperproperties. The model checking problem of \sohyperltl\ is, in general, undecidable. For the expressive fragment
where second-order quantification is restricted to smallest and largest sets, we present an approximate model-checking algorithm that computes increasingly precise under- and overapproximations of the quantified sets, based on fixpoint iteration and automata learning. We report on encouraging experimental results with our model-checking algorithm, which we implemented in the tool~\texttt{HySO}. 

\end{abstract}

\section{Introduction}\label{sec:intro}

About a decade ago, Clarkson and Schneider coined the term \emph{hyperproperties}~\cite{DBLP:journals/jcs/ClarksonS10} for the rich class of system requirements that relate multiple computations. In their definition, hyperproperties generalize trace properties, which are sets of traces, to \emph{sets of} sets of traces. This covers a wide range of requirements, from information-flow security policies to epistemic properties describing the knowledge of agents in a distributed system. Missing from Clarkson and Schneider's original theory was, however, a concrete specification language that could express customized hyperproperties for specific applications and serve as the common semantic foundation for different verification methods.

A first milestone towards such a language was the introduction of the temporal logic HyperLTL~\cite{ClarksonFKMRS14}. HyperLTL extends linear-time temporal logic (LTL) with quantification over traces. Suppose, for example, that an agent $i$ in a distributed system observes only a subset of the system variables. The agent \emph{knows} that some LTL formula $\varphi$ is true on some trace $\pi$ iff $\varphi$ holds on \emph{all} traces $\pi'$ that agent $i$ cannot distinguish from $\pi$. If we denote the indistinguishability of $\pi$ and $\pi'$ by $\pi \sim_i \pi'$, then the property that \emph{there exists a trace $\pi$ where agent~$i$ knows $\varphi$} can be expressed as the HyperLTL formula
\begin{align*}
    \exists \pi. \forall \pi'\ldot \pi \sim_i \pi'  \rightarrow \varphi(\pi'),
\end{align*}
where we write $\varphi(\pi')$ to denote that the trace property $\varphi$ holds on trace $\pi'$.

While HyperLTL and its variations have found many applications~\cite{FinkbeinerRS15,10.1145/3127041.3127058,DimitrovaFT20}, the expressiveness of these logics is limited, leaving many widely used hyperproperties out of reach. A prominent example is \emph{common knowledge}, which is used in distributed applications to ensure simultaneous action~\cite{DBLP:books/mit/FHMV1995,HM}. Common knowledge in a group of agents means that the agents not only know \emph{individually} that some condition $\varphi$ is true, but that this knowledge is ``common" to the group in the sense that each agent \emph{knows} that every agent \emph{knows} that $\varphi$ is true; on top of that, each agent in the group \emph{knows} that every agent \emph{knows} that every agent \emph{knows} that $\varphi$ is true; and so on, forming an infinite chain of knowledge. 

The fundamental limitation of HyperLTL that makes it impossible to express properties like common knowledge is that the logic is restricted to 
\emph{first-order quantification}. HyperLTL, then, cannot reason about sets of traces directly, but must always do so by referring to individual traces that are chosen existentially or universally from the full set of traces. For the specification of an agent's individual knowledge, where we are only interested in the (non-)existence of a single trace that is indistinguishable and that violates $\varphi$, this is sufficient; however, expressing an infinite chain, as needed for common knowledge, is impossible.

In this paper, we introduce \sohyperltl{}, a temporal logic for hyperproperties with \emph{second-order quantification} over traces. In \sohyperltl{}, the existence of a trace $\pi$ where the condition $\varphi$ is common knowledge can be expressed as the following formula (using slightly simplified syntax):
\begin{align*}
    \exists \pi.\, \exists X \ldot\ \pi \in X \land \Big( \forall \pi' \in X.\, \forall \pi''.\ \big(\bigvee_{i=1}^n \pi' \sim_i \pi'' \big) \rightarrow \pi'' \in X \Big)\, \land\, \forall \pi' \in X\ldot \varphi({\pi'}).
\end{align*}
The second-order quantifier $\exists X$ postulates the existence of a set $X$ of traces that (1) contains $\pi$; that (2) is closed under the observations of each agent, i.e., for every trace $\pi'$ already in $X$, 
all other traces $\pi''$ that some agent~$i$ cannot distinguish from $\pi'$ are also in $X$; and that (3) only contains traces that satisfy $\varphi$.
The existence of $X$ is a necessary and sufficient condition for $\varphi$ being common knowledge on $\pi$. In the paper, we show that \sohyperltl{} is an elegant specification language for many hyperproperties of interest that cannot be expressed in HyperLTL, including, in addition to epistemic properties like common knowledge, also Mazurkiewicz trace theory and asynchronous hyperproperties.

The model checking problem for \sohyperltl{} is much more difficult than for HyperLTL. A HyperLTL formula can be checked by translating the LTL subformula into an automaton and then applying a series of automata transformations, such as self-composition to generate multiple traces, projection for existential quantification, and complementation for negation~\cite{FinkbeinerRS15,BeutnerF23}. For \sohyperltl{}, the model checking problem is, in general, undecidable. We introduce a method that nevertheless obtains sound results by over- and underapproximating the quantified sets of traces. For this purpose, we study \fphyperltl{}, a fragment of \sohyperltl{}, in which we restrict second-order quantification to the smallest or largest set satisfying some property. For example, to check common knowledge, it suffices to consider the \emph{smallest} set $X$ that is closed under the observations of all agents.
This smallest set $X$ is defined by the (monotone) fixpoint operation that adds, in each step, all traces that are indistinguishable to some trace already in $X$.

We develop an approximate model checking algorithm for \fphyperltl{} that uses bidirectional inference to deduce lower and upper bounds on second-order variables, interposed with first-order model checking in the style of HyperLTL. Our procedure is parametric in an oracle that provides (increasingly precise) lower and upper bounds. In the paper, we realize the oracles with \emph{fixpoint iteration} for underapproximations of the sets of traces assigned to the second-order variables, and \emph{automata learning} for overapproximations.
We report on encouraging experimental results with our model-checking algorithm, which has been implemented in a tool called \texttt{HySO}.

\section{Preliminaries}\label{sec:prelim}

For $n \in \nat$ we define $[n] := \{1, \ldots, n\}$.
We assume that $\ap$ is a finite set of atomic propositions and define $\Sigma := 2^\ap$.
For $t \in \Sigma^\omega$ and $i \in \nat$ define $t(i) \in \Sigma$ as the $i$th element in~$t$ (starting with the $0$th); and $t[i,\infty]$ for the infinite suffix starting at position~$i$.
For traces $t_1, \ldots, t_n \in \Sigma^\omega$ we write $\mathit{zip}(t_1, \ldots, t_n) \in (\Sigma^n)^\omega$ for the pointwise zipping of the traces, i.e., $\mathit{zip}(t_1, \ldots, t_n)(i) := (t_1(i), \ldots, t_n(i))$.

\paragraph{Transition systems.}\label{prelim:mealy:machines}
A \emph{transition system} is a tuple $\calT = (S, S_0, \kappa, L)$ where $S$ is a set of states, $S_0 \subseteq S$ is a set of initial states, $\kappa \subseteq S \times S$ is a transition relation, and $L : S \to \Sigma$ is a labeling function. 
A path in $\calT$ is an infinite state sequence $s_0s_1s_2 \cdots \in S^\omega$, s.t., $s_0 \in S_0$, and $(s_i, s_{i+1}) \in \kappa$ for all $i$.
The associated trace is given by $L(s_0)L(s_1)L(s_2) \cdots \in \Sigma^\omega$ and $\traces(\calT) \subseteq \Sigma^\omega$ denotes all traces of $\calT$.

\paragraph{Automata.}
A \emph{non-deterministic B{\"u}chi automaton} (NBA) \cite{Buechi62Decision} is a tuple $\mathcal{A}= (\Sigma, Q, q_0, \delta, F)$ where $\Sigma$ is a finite {alphabet}, $Q$ is a finite set of states, $Q_0\subseteq Q$ is the set of {initial states}, $F\subseteq Q$ is a set of {accepting states}, and $\delta : Q\times \Sigma \to 2^Q$ is the {transition function}. 
A run on a word $u \in \Sigma^\omega$ is an infinite sequence of states $q_0q_1q_2 \cdots \in Q^\omega$ such that $q_0 \in Q_0$ and for every $i \in \nat$, $q_{i+1} \in \delta(q_i, u(i))$.
The run is accepting if it visits states in $F$ infinitely many times, and we define the language of $\calA$, denoted $\mathcal{L}(\mathcal{A}) \subseteq \Sigma^\omega$, as all infinite words on which $\calA$ has an accepting run.

\paragraph{HyperLTL.}
HyperLTL \cite{ClarksonFKMRS14} is one of the most studied temporal logics for the specification of hyperproperties. 
We assume that $\tracevars$ is a fixed set of trace variables. 
For the most part, we use variations of $\pi$ (e.g., $\pi, \pi', \pi_1, \ldots$) to denote trace variables.
HyperLTL formulas are then generated by the grammar
\begin{align*}
\varphi &:= \quant \pi \ldot \varphi \mid \psi \\
\psi &:= a_\pi \mid \neg \psi \mid \psi \land \psi \mid \X \psi \mid \psi \U \psi 
\end{align*}
where $a \in \ap$ is an atomic proposition, $\pi \in \tracevars{}$ is a trace variable, $\quant \in \{\forall, \exists\}$ is a quantifier, and $\LTLnext$ and $\LTLuntil$ are the temporal operators \emph{next} and \emph{until}.

The semantics of HyperLTL is given with respect to a \emph{trace assignment} $\Pi$, which is a partial mapping $\Pi : \tracevars \rightharpoonup \Sigma^\omega$ that maps trace variables to traces. 
Given $\pi \in \tracevars$ and $t \in \Sigma^\omega$ we define $\traceassign[\pi \mapsto t]$ as the updated assignment that maps $\pi$ to $t$.
For $i \in \nat$ we define $\Pi[i, \infty]$ as the trace assignment defined by $\Pi[i, \infty](\pi) := \Pi(\pi)[i, \infty]$, i.e., we (synchronously) progress all traces by $i$ steps.
For quantifier-free formulas $\psi$ we follow the LTL semantics and define
\begin{align*}
	\Pi &\models  a_\pi &\text{iff} \quad  &a \in \Pi(\pi)(0)\\
	\Pi &\models  \neg \psi &\text{iff} \quad & \Pi \not\models  \psi \\
	\Pi &\models  \psi_1 \land \psi_2 &\text{iff} \quad  &\Pi \models \psi_1 \text{ and }  \Pi \models  \psi_2\\
	\Pi &\models  \X  \psi &\text{iff} \quad & \Pi[1,\infty] \models \psi \\
	\Pi &\models  \psi_1 \U \psi_2 &\text{iff} \quad & \exists i \in \nat \ldot \Pi[i,\infty]\models  \psi_2 \text{ and } \forall j < i\ldot  \Pi[j, \infty] \models  \psi_1\, .
\end{align*}
The indexed atomic propositions refer to a specific path in $\Pi$, i.e., $a_\pi$ holds iff $a$ holds on the trace bound to $\pi$.
 Quantifiers range over system traces: 
%
\begin{align*}
    \Pi \models_\calT \psi  \text{ iff }  \Pi \models\psi \quad\quad \text{and} \quad\quad \Pi \models_\calT  \quant \pi \ldot \varphi \text{ iff } \quant t \in \traces(\calT) \ldot \Pi[\pi \mapsto t] \models  \varphi\, .
\end{align*}
We write $\calT \models \varphi$ if $\emptyset \models_\calT \varphi$ where $\emptyset$ denotes the empty trace assignment.

\paragraph{HyperQPTL.} 
HyperQPTL~\cite{DBLP:phd/dnb/Rabe16} adds -- on top of the trace quantification of HyperLTL -- also propositional quantification (analogous to the propositional quantification that QPTL~\cite{qptl} adds on top of LTL). 
For example, HyperQPTL can express a promptness property which states that there must exist a bound (which is common among all traces), up to which an event must have happened.
We can express this as $\exists q. \forall \pi\ldot \F q \land (\neg q) \LTLuntil a_\pi$ which states that there exists an evaluation of proposition $q$ such that (1) $q$ holds at least once, and (2) for all traces $\pi$, $a$ holds on $\pi$ before the first occurrence of $q$.
See \cite{BeutnerF23} for details.

\section{Second-Order HyperLTL}\label{sec:second:order:hyperltl}

The (first-order) trace quantification in HyperLTL ranges over the set of all system traces; we thus cannot reason about arbitrary sets of traces as required for, e.g., common knowledge. 
We introduce a second-order extension of HyperLTL by introducing second-order variables (ranging over sets of traces) and allowing quantification over traces from any such set.
We present two variants of our logic that differ in the way quantification is resolved. 
In \sohyperltl{}, we quantify over arbitrary sets of traces.
While this yields a powerful and intuitive logic, second-order quantification is inherently non-constructive. 
During model checking, there thus does not exist an efficient way to even approximate possible witnesses for the sets of traces. 
To solve this quandary, we restrict \sohyperltl{} to \fphyperltl{}, where we instead quantify over sets of traces that satisfy some minimality or maximality constraint.
This allows for large fragments of \fphyperltl{} that admit algorithmic approximations to its model checking (by, e.g., using known techniques from fixpoint computations~\cite{tarski1955lattice,Winskel93}). 

\subsection{\oldsohyperltl{}}\label{sec:sohyperltl}
Alongside the set $\tracevars{}$ of trace variables, we use a set $\sovars{}$ of second-order variables (which we, for the most part, denote with capital letters $X, Y, ...$). 
We assume that there is a special variable $\systemvar \in \sovars{}$ that refers to the set of traces of the given system at hand, and a variable $\allvar \in \sovars{}$ that refers to the set of all traces.
We define the \sohyperltl{} syntax by the following grammar:
\begin{align*}
\varphi &:= \quant \pi \in X \ldot \varphi \mid \quant X\ldot \varphi \mid \psi \\ 
\psi &:= a_\pi \mid \neg \psi \mid \psi \land \psi \mid \X \psi \mid \psi \U \psi  
\end{align*}
where $a \in \ap{}$ is an atomic proposition, $\pi \in \tracevars{}$ is a trace variable, $X \in \sovars{}$ is a second-order variable, and $\quant \in \{\forall, \exists\}$ is a quantifier.
We also consider the usual derived Boolean constants ($\true$, $\false$) and connectives ($\lor$, $\rightarrow$, $\leftrightarrow$) as well as the temporal operators \emph{eventually} ($\F \psi := \true \U \psi$) and \emph{globally} ($\G \psi := \neg \F \neg \psi$).
Given a set of atomic propositions $P \subseteq \ap$ and two trace variables $\pi, \pi'$, we abbreviate $\pi =_P \pi' := \bigwedge_{a \in P} (a_\pi \leftrightarrow a_{\pi'})$.

\subsubsection*{Semantics.}
Apart from a trace assignment $\Pi$ (as in the semantics of HyperLTL), we maintain a second-order assignment $\Delta : \sovars{} \rightharpoonup 2^{\Sigma^\omega}$ mapping second-order variables to \emph{sets of traces}. 
Given $X \in \sovars$ and $A \subseteq \Sigma^\omega$ we define the updated assignment $\Delta[X \mapsto A]$ as expected.
Quantifier-free formulas $\psi$ are then evaluated in a fixed trace assignment as for HyperLTL (cf.~\Cref{sec:prelim}).
For the quantifier prefix we define:
\begin{align*}
     \Pi, \Delta &\models\psi  &\text{iff} \quad &\Pi \models\psi\\
	\Pi, \Delta &\models  \quant \pi \in X \ldot \varphi &\text{iff} \quad &\quant t \in \Delta(X) \ldot \Pi[\pi \mapsto t], \Delta \models  \varphi\\
	\Pi, \Delta &\models  \quant X \ldot \varphi &\text{iff} \quad &\quant A \subseteq \Sigma^\omega \ldot \Pi, \Delta[X \mapsto A] \models  \varphi 
\end{align*}
Second-order quantification updates $\Delta$ with a set of traces, and first-order quantification updates $\Pi$ by quantifying over traces within the set defined by $\Delta$.

Initially, we evaluate a formula in the empty trace assignment and fix the valuation of the special second-order variable $\systemvar$ to be the set of all system traces and $\allvar$ to be the set of all traces.
That is, given a system $\calT$ and \sohyperltl{} formula $\varphi$, we say that $\calT$ satisfies $\varphi$, written $\calT \models \varphi$, if $\emptyset, [\systemvar \mapsto \traces(\calT), \allvar \mapsto \Sigma^\omega] \models \varphi$, where we write $\emptyset$ for the empty trace assignment.
The model-checking problem for \sohyperltl{} is checking whether $\calT \models \varphi$ holds.

\sohyperltl{} naturally generalizes HyperLTL by adding second-order quantification. 
As sets range over \emph{arbitrary} traces, \sohyperltl{} also subsumes the~more powerful logic HyperQPTL.  
The proof of~\Cref{lemma:ltlto2} is given in~\Cref{app:qptl}. 

\begin{restatable}{lemma}{hyperltlIntoSoHyper}\label{lemma:ltlto2}
    \sohyperltl{} subsumes \hyperqptl{} (and thus also \hyperltl{}).
\end{restatable}

\subsubsection*{Syntactic Sugar.}
In \sohyperltl{}, we can quantify over traces within a second-order variable, but we cannot state, within the body of the formula, that some path is a member of some second-order variable.
For that, we define $\pi \modin X$ (as an atom within the body) as syntactic sugar for $\exists \pi' \in X. \G(\pi' =_\ap \pi)$, i.e., $\pi$ is in $X$ if there exists some trace in $X$ that agrees with $\pi$ on all propositions. 
Note that we can only use $\pi \modin X$ \emph{outside} of the scope of any temporal operators; this ensures that we can bring the resulting formula into a form that conforms to the \sohyperltl{} syntax.

\subsection{\oldfphyperltl{}}\label{sec:fphyperltl}

The semantics of \sohyperltl{} quantifies over arbitrary sets of traces, making even approximations to its semantics challenging. 
We propose \fphyperltl{} as a restriction that only quantifies over sets that are subject to an additional minimality or maximality constraint. 
For large classes of formulas, we show that this admits effective model-checking approximations.
We define \fphyperltl{} by the following grammar:
\begin{align*}
\varphi &:= \quant\,\pi \in X \ldot \varphi \mid \quant\,  (X, \fptype{}, \varphi) \ldot \varphi \mid \psi\\
\psi &:= a_\pi \mid \neg \psi \mid \psi \land \psi \mid \X \psi \mid \psi \U \psi  
\end{align*}
where $a \in \ap$, $\pi \in \tracevars{}$, $X \in \sovars$, $\quant \in \{\forall, \exists\}$, and $\fptype{} \in \{\curlywedge, \curlyvee\}$ determines if we consider smallest ($\smallestSet$) or largest ($\largestSet$) sets.  
For example, the formula $\exists\, (X, \smallestSet, \varphi_1) \ldot \varphi_2$ holds if there exists some set of traces $X$, that satisfies both $\varphi_1$ and $\varphi_2$, and is \emph{a} smallest set that satisfies~$\varphi_1$.
Such minimality and maximality constraints with respect to a (hyper)property arise naturally in many properties.
Examples include common knowledge (cf.~\Cref{sec:running:example}), asynchronous hyperproperties (cf.~\Cref{sec:asynchronous:hyperproperties}), and causality in reactive systems~\cite{DBLP:conf/atva/CoenenFFHMS22,DBLP:conf/cav/CoenenDFFHHMS22}.

\subsubsection*{Semantics.}
For path formulas, the semantics of \fphyperltl{} is defined analogously to that of \sohyperltl{} and HyperLTL.
For the quantifier prefix we define:
\begin{align*}
    \Pi, \Delta &\models  \psi &\text{iff} \quad &\Pi \models \psi\\
	\Pi, \Delta &\models  \quant \pi \in X \ldot \varphi &\text{iff} \quad &\quant t \in \Delta(X) \ldot \Pi[\pi \mapsto t], \Delta \models  \varphi\\
	\Pi, \Delta &\models  \quant  (X, \fptype{}, \varphi_1) \ldot \varphi_2 &\text{iff} \quad &\quant A \in \mathit{sol}(\Pi, \Delta, (X, \fptype{}, \varphi_1)) \ldot \Pi, \Delta[X \mapsto A] \models  \varphi_2
\end{align*}
where $\mathit{sol}(\Pi, \Delta, (X, \fptype{}, \varphi_1))$ denotes all solutions to the minimality/maximality condition given by $\varphi_1$, which we define by mutual recursion as follows:\\
\scalebox{0.95}{\parbox{\linewidth}{
\begin{align*}
    \mathit{sol}(\Pi, \Delta, (X, \smallestSet, \varphi))&:= \{
    A \subseteq \Sigma^\omega \mid \Pi, \Delta[X \mapsto A] \models \varphi \land \forall A' \subsetneq A\ldot \Pi, \Delta[X \mapsto A'] \not\models \varphi
    \}\\
    \mathit{sol}(\Pi, \Delta, (X, \largestSet, \varphi))&:= \{
    A \subseteq \Sigma^\omega \mid \Pi, \Delta[X \mapsto A] \models \varphi \land \forall A' \supsetneq A\ldot  \Pi, \Delta[X \mapsto A'] \not\models \varphi
    \}
\end{align*}
}}\\
A set $A$ satisfies the minimality/maximality constraint if it satisfies $\varphi$ and is a least (in case $\fptype{} = \smallestSet$) or greatest (in case $\fptype{} = \largestSet$) set that satisfies $\varphi$.

Note that $\mathit{sol}(\Pi, \Delta, (X,\fptype{}, \varphi))$ can contain multiple sets or no set at all, i.e., there may not exists a unique least or greatest set that satisfies $\varphi$.
In \fphyperltl{}, we therefore add an additional quantification over the set of all solutions to the minimality/maximality constraint. 
When discussing our model checking approximation algorithm, we present a (syntactic) restriction on $\varphi$ which guarantees that $\mathit{sol}(\Pi, \Delta, (X, \fptype{}, \varphi))$ contains a unique element (i.e., is a singleton set). 
Moreover, our restriction allows us to employ fixpoint techniques to find approximations to this unique solution.
In case the solution for $(X,\fptype{}, \varphi)$ is unique, we often omit the leading quantifier and simply write $(X, \fptype{}, \varphi)$ instead of $\quant (X,\fptype{}, \varphi)$.

As we can encode the minimality/maximality constraints of \fphyperltl{} in \sohyperltl{} (see~\Cref{app:fptoso}), we have the following:

\begin{restatable}{proposition}{fpToSo}\label{prop:fptoso}
Any \fphyperltl{} formula $\varphi$ can be effectively translated into an \sohyperltl{} formula $\varphi'$ such that for all transition systems $\calT$ we have $\calT \models \varphi$ iff $\calT \models \varphi'$.
\end{restatable}

\begin{figure}[t]
     \centering
    \begin{minipage}{0.49\textwidth}
    \centering
    \resizebox{.65\linewidth}{!}{
    \tikzstyle{state}=[draw, circle, fill=none, minimum width=.8cm, 
minimum height = .8cm, thick]

\begin{tikzpicture}[->,>=stealth',shorten >= 1pt,auto]

\node[state] (s1) [] {%
   $b$
 };

\node[state] (s0) [above left = 0.2 and 1 of s1]{%
   $a$
};

\node (left) [left = 0.5 of s0, draw=none]{};

\node[state] (s2) [above right = 0.2 and 1 of s1] {%
    $d$
 };

  \node[state] (s3) [below= .7 of s1] {%
    $c$
 };




\path (left) edge (s0)
      (s0) edge[thick, bend left=0, align=center] node[above, sloped] {%
      } (s1)
      (s0) edge[loop above, thick,  align=center] node[above,sloped] {%
      } (0)
      (s0) edge[thick, bend right, align=center] node[sloped,below] {%
      } (s3)
      (s0) edge[thick, bend left, align=center] node[sloped,below] {%
      } (s2)
      (s1) edge[thick, bend right=0, align=center] node[sloped,below] {%
      } (s2)
      (s1) edge[thick, bend right=0, align=center] node[sloped,below] {%
      } (s3)
      (s2) edge[loop above, thick, align=center] node[sloped,below] {%
      } (s2)
      (s3) edge[loop below, thick, align=center] node[sloped,below] {%
      } (s3)
      (s3) edge[bend right, thick,align=center] node[sloped,below] {%
      } (s2)
      ;
\end{tikzpicture}
    }
    \end{minipage}
    \begin{minipage}{.49\textwidth}
    \centering
    \begin{small}
   \begin{align*}
   \pi =&\ a^nd^\omega \\
    K_2 (\pi) = &\ a^{n-1}bd^\omega\\
    K_1K_2(\pi) = &\ a^{n-1}cd^\omega\\
    K_2K_1K_2(\pi) = &\ a^{n-2}bcd^\omega\\
   \ldots\\
    K_1K_2\ldots K_2(\pi) =&\ ac^{n-1}d^\omega
   \end{align*}
   \end{small}
   \end{minipage}
    \caption{Left: An example for a multi-agent system with two agents, where agent~1 observes $a$ and $d$, and agent 2 observes $c$ and $d$. Right: The iterative construction of the traces to be considered for common knowledge starting with $a^nd^\omega$.}
        \label{fig:common:knowledge}
\end{figure}

\subsection{Common Knowledge in Multi-Agent Systems} \label{sec:running:example}
To explain common knowledge, we use a variation of an example from~\cite{Meyden98}, and encode it in \fphyperltl{}. 
\Cref{fig:common:knowledge}(left) shows a transition system of a distributed system with two agents, agent~$1$ and agent $2$.
Agent~$1$ observes variables $a$ and $d$, whereas agent $2$ observes $c$ and~$d$.
The property of interest is \emph{starting from the trace $\pi = a^n d^\omega$ for some fixed $ n > 1$, is it common knowledge for the two agents that $a$ holds in the second step}.
It is trivial to see that $\LTLnext a$ holds on~$\pi$.
However, for common knowledge, we consider the (possibly) infinite chain of observationally equivalent traces. 
For example, agent $2$ cannot distinguish the traces $a^nd^\omega$ and $a^{n-1}bd^\omega$. Therefore, agent $2$ only knows that $\LTLnext a$ holds on $\pi$ if it also holds on $ \pi' = a^{n-1}bd^\omega$.
For common knowledge, agent $1$ also has to know that agent~$2$ knows $\LTLnext a$, which means that for all traces that are indistinguishable from $\pi$ or~$\pi'$ for agent $1$,  
$\LTLnext a$ has to hold.
This adds $\pi'' =  a^{n-1}cd^\omega$ to the set of traces to verify $\LTLnext a$ against.
This chain of reasoning continues as shown in \Cref{fig:common:knowledge}(right). In the last step we add $ac^{n-1}d^\omega$ to the set of indistinguishable traces, concluding that  $\LTLnext a$ is not common knowledge.

The following \fphyperltl{} formula specifies the property stated above.
The abbreviation $\mathit{obs}(\pi_1, \pi_2) :=  \G (\pi_1 =_{\{a, d\}} \pi_2) \vee \G(\pi_1 =_{\{c, d\}} \pi_2)$ denotes that $\pi_1$ and $\pi_2$ are observationally equivalent for either agent 1 or agent 2.
\begin{align*}
    &\forall \pi \in \systemvar. \big( \bigwedge_{i=0}^{n-1} \X^i a_\pi \land \X^{n}\G d_\pi\big) \to \\
    &\quad\Big(X, \smallestSet,  \pi \modin X  \land \big(\forall \pi_1 \in X. \forall \pi_2 \in \systemvar{}\ldot \mathit{obs}(\pi_1,\pi_2) \rightarrow \pi_2 \modin X  \big)\Big)\ldot \forall \pi' \in X. \LTLnext a_{\pi'}
\end{align*}

For a trace $\pi$ of the form $\pi = a^n d^\omega$, the set $X$ represents the \emph{common knowledge set} on~$\pi$.
This set $X$ is the smallest set that (1) contains $\pi$ (expressed using our syntactic sugar $\modin$); and (2) is closed under observations by either agent, i.e., if we find some $\pi_1 \in X$ and some system trace $\pi_2$ that are observationally equivalent, $\pi_2$ should also be in $X$.  
Note that this set is unique (due to the minimality restriction), so we do not quantify it explicitly. 
Lastly, we require that all traces in $X$ satisfy the property $\LTLnext a$. 
All sets that satisfy this formula would also include the trace $ac^{n-1}d^\omega$, and therefore no such $X$ exists; thus, we can conclude that starting from trace $a^nd^\omega$, it is \emph{not} common knowledge that $\LTLnext a$ holds. 

On the other hand, it \emph{is} common knowledge that $a$ holds in the \emph{first} step (cf.~\Cref{sec:implementation}).

\subsection{\oldsohyperltl{} Model Checking}\label{sec:sohyperltl_mc}

As \sohyperltl{} and \fphyperltl{} allow quantification over arbitrary sets of traces, we can encode the satisfiability of \HyperQPTL{} (i.e., the question of whether some set of traces satisfies a formula) within their model-checking problem; rendering the model-checking problem highly undecidable~\cite{FortinKT021}, even for very simple formulas \cite{BeutnerCFHK22}.

\begin{restatable}{proposition}{hyperSatToSo}\label{prop:hyperltl_sat_in_sohyperltl_mc}\label{corr:sohyper_hardness}
    For any \HyperQPTL{} formula $\varphi$ there exists a \sohyperltl{} formula $\varphi'$ such that $\varphi$ is satisfiable iff $\varphi'$ holds on some arbitrary transition system.
    The model-checking problem of \sohyperltl{} is thus highly undecidable ($\Sigma_1^1$-hard).
\end{restatable}
\begin{proof}
    Let $\varphi'$ be the \sohyperltl{} formula obtained from $\varphi$ by replacing each \HyperQPTL{} trace quantifier $\quant \pi$ with the \sohyperltl{} quantifier $\quant \pi \in X$, and each propositional quantifier $\quant q$ with $\quant \pi_q \in \allvar$ for some fresh trace variable $\pi_q$.
    In the body, we replace each propositional variable $q$ with $a_{\pi_q}$ for some fixed proposition $a \in \ap$.
    Then, $\varphi$ is satisfiable iff the \sohyperltl{} formula $\exists X. \varphi'$ holds in some arbitrary system. 
    \qed
\end{proof}

\fphyperltl{} cannot express \HyperQPTL{} satisfiability directly. If there exists a model of a \HyperQPTL{} formula, there may not exist a least one. 
However, model checking of \fphyperltl{} is also highly undecidable.

 \begin{restatable}{proposition}{fphyperUndec}\label{lem:fphyper_hardness}
    The model-checking problem of \fphyperltl{} is $\Sigma_1^1$-hard.
\end{restatable}
\begin{proof}[Sketch]
    We can encode the existence of a \emph{recurrent} computation of a Turing machine, which is known to be $\Sigma_1^1$-hard \cite{AlurH94}.\qed
\end{proof}

Conversely, the \emph{existential} fragment of \sohyperltl{} can be encoded back into \hyperqptl{} satisfiability:

\begin{restatable}{proposition}{sohyperToSat}
    Let $\varphi$ be a \sohyperltl{} formula that uses only existential second-order quantification and $\calT$ be any system.
    We can effectively construct a formula $\varphi'$ in \hyperqptl{} such that $\calT \models \varphi$ iff $\varphi'$ is satisfiable.
\end{restatable}

Lastly, we present some easy fragments of \sohyperltl{} for which the model-checking problem is decidable. 
Here we write $\exists^* X$ (resp.~$\forall^* X$) for some sequence of existentially (resp.~universally) quantified \emph{second-order} variables and $\exists^* \pi$ (resp.~$\forall^* \pi$) for some sequence of existentially (resp.~universally) quantified \emph{first-order} variables. 
For example, $\exists^* X \forall^*\pi$ captures all formulas of the form $\exists X_1, \ldots X_n. \forall \pi_1, \ldots, \pi_m. \psi$ where $\psi$ is quantifier-free. 

\begin{restatable}{proposition}{easyfrags}
\label{prop:easyfragments}
The model-checking problem of \sohyperltl{} is decidable for the fragments: $\exists^* X \forall^*\pi$, $\forall^* X \forall^*\pi$, $\exists^* X \exists^*\pi$, $\forall^* X \exists^*\pi$, $\exists X.\exists^* \pi\in X \forall^*\pi'\in X$.
\end{restatable}

See \Cref{app:proofs} for the full proofs of the propositions above. 

\section{Expressiveness of \oldsohyperltl{}}\label{sec:examples}
In this section, we point to existing logics that can naturally be encoded within our second-order hyperlogics \sohyperltl{} and \fphyperltl{}.

\subsection{\oldsohyperltl{} and LTL$_{\calK, \sfC}$}\label{sec:ltlK}

LTL$_\calK$ extends LTL with the knowledge operator $\calK$.
For some subset of agents $A$, the formula $\calK_A \psi$ holds in timestep $i$, if $\psi$ holds on all traces equivalent to some agent in $A$ up to timestep $i$.
See \Cref{app:ltlk} for detailed semantics. 
LTL$_\calK$ and HyperCTL$^*$ have incomparable expressiveness \cite{DBLP:conf/fossacs/BozzelliMP15} but the knowledge operator $\calK$ can be encoded by either adding a linear past operator \cite{DBLP:conf/fossacs/BozzelliMP15} or by adding propositional quantification (as in \HyperQPTL{}) \cite{DBLP:phd/dnb/Rabe16}.

Using \fphyperltl{} we can encode LTL$_{\calK, \sfC}$, featuring the knowledge operator $\calK$ \emph{and} the common knowledge operator $\sfC$ (which requires that $\psi$ holds on the closure set of equivalent traces, up to the current timepoint) \cite{DBLP:conf/spin/HoekW02}.
Note that LTL$_{\calK, \sfC}$ is not encodable by only adding propositional quantification  or the linear past operator.

\begin{restatable}{proposition}{kltl}
    For every \emph{LTL$_{\calK, \sfC}$} formula $\varphi$ there exists an \fphyperltl{} formula $\varphi'$ such that for any system $\calT$ we have $\calT \models_{\text{LTL$_{\calK, \sfC}$}} \varphi$ iff  $\calT \models \varphi'$.
\end{restatable}

\begin{proof}[Sketch]
    We follow the intuition discussed in~\Cref{sec:running:example}. For each occurrence of a knowledge operator in $\{\calK, \sfC \}$, we introduce a second-order set that collects all equivalent traces, and a fresh trace variable that keeps track on the points in time with respect to which we need to compare. We then inductively construct a \fphyperltl{} formula that captures all the knowledge and common-knowledge sets. For more details see~\Cref{app:ltlk}. \qed
\end{proof}

\subsection{\oldsohyperltl{} and Asynchronous Hyperproperties}\label{sec:asynchronous:hyperproperties}

Most existing hyperlogics (including \sohyperltl{}) traverse the traces of a system \emph{synchronously}.
However, in many cases such a synchronous traversal is too restricting and we need to compare traces asynchronously.
As an example, consider \emph{observational determinism} (OD), which we can express in HyperLTL as $\varphi_\mathit{OD} := \forall \pi_1. \forall \pi_2. \G(o_{\pi_1} \leftrightarrow o_{\pi_2})$.
The formula states that the output of a system is identical across all traces and so (trivially) no information about high-security inputs is leaked. 
In most systems encountered in practice, this synchronous formula is violated, as the exact timing between updates to $o$ might differ by a few steps (see \Cref{app:asynchronous:hyperproperties} for some examples). 
However, assuming that an attacker only has access to the memory footprint and not a timing channel, we would only like to check that all traces are \emph{stutter} equivalent (with respect to~$o$).

A range of extensions to existing hyperlogics has been proposed to reason about such asynchronous hyperproperties~\cite{BaumeisterCBFS21,BozzelliPS21,GutsfeldMO21,BeutnerF21,BeutnerF23LMCS}. We consider AHLTL~\cite{BaumeisterCBFS21}.
An AHLTL formula has the form $\quant_1 \pi_1, \ldots , \quant_n \pi_m. \mathbf{E} \ldot \psi$ where $\psi$ is a qunatifier-free HyperLTL formula.
The initial trace quantifier prefix is handled as in HyperLTL.
However, different from HyperLTL, a trace assignment $[\pi_1 \mapsto t_1, \ldots, \pi_n \mapsto t_n]$ satisfies $\mathbf{E} \ldot \psi$ if there exist stuttered traces $t_1', \ldots, t_n'$ of $t_1, \ldots, t_n$ such that $[\pi_1 \mapsto t_1', \ldots, \pi_n \mapsto t_n'] \models \psi$.
We write $\calT \models_{\mathit{AHLTL}} \varphi$ if a system $\calT$ satisfies the AHLTL formula $\varphi$.
Using this quantification over stutterings we can, for example, express an asynchronous version of observational determinism as $\forall \pi_1. \forall \pi_2. \mathbf{E} \ldot \G(o_{\pi_1} \leftrightarrow o_{\pi_2})$ stating that every two traces can be aligned such that they (globally) agree on $o$.
Despite the fact that \fphyperltl{} is itself synchronous, we can use second-order quantification to encode asynchronous hyperproperties, as we state in the following proposition. 

\begin{proposition}
    For any AHLTL formula $\varphi$ there exists a \fphyperltl{} formula $\varphi'$ such that for any system $\calT$ we have $\calT \models_{\mathit{AHLTL}} \varphi$ iff  $\calT \models \varphi'$.
\end{proposition}
\begin{proof}
    Assume that $\varphi = \quant_1 \pi_1, \ldots, \quant_n \pi_n. \mathbf{E} \ldot \psi$ is the given AHLTL formula.
    For each $i \in [n]$ we define a formula $\varphi_i$ as follows
    \begin{align*}
     \forall \pi_1 \in X_i. \forall \pi_2 \in \allvar{}\ldot
    \Big( \big( \pi_1 =_\ap \pi_2 \big) \U \big( \G \bigwedge_{a \in \ap} a_{\pi_1} \leftrightarrow \X a_{\pi_2} \big)
    \Big) \rightarrow \pi_2 \modin X_i
\end{align*}
The formula asserts that the set of traces bound to $X_i$ is closed under stuttering, i.e., if we start from any trace in $X_i$ and stutter it once (at some arbitrary position) we again end up in $X_i$.
Using the formulas $\varphi_i$, we then construct a \fphyperltl{} formula that is equivalent to $\varphi$ as follows
\begin{align*}
    \varphi' := \,&\quant_1 \pi_1 \in \systemvar{}, \ldots, \quant_n \pi_n \in \systemvar{}. (X_1, \smallestSet, \pi_1 \modin X_1 \land \varphi_1) \cdots (X_n, \smallestSet, \pi_n \modin X_n \land \varphi_n) \\
    &\quad\quad\exists \pi_1' \in X_1, \ldots, \exists \pi_n' \in X_n. \psi[\pi_1'/\pi_1, \ldots, \pi_n'/\pi_n]
\end{align*}
We first mimic the quantification in $\varphi$ and, for each trace $\pi_i$, construct a least set $X_i$ that contains $\pi_i$ and is closed under stuttering (thus describing exactly the set of all stuttering of $\pi_i$).
Finally, we assert that there are traces $\pi_1', \ldots, \pi_n'$ with $\pi_i' \in X_i$ (so $\pi_i'$ is a stuttering of $\pi_i$) such that $\pi_1', \ldots, \pi_n'$ satisfy $\psi$.
It is easy to see that $\calT \models_{\mathit{AHLTL}} \varphi$ iff $\calT \models \varphi'$ holds for all systems.
\qed
\end{proof}

\fphyperltl{} captures all properties expressible in AHLTL.
In particular, our approximate model-checking algorithm for \fphyperltl{} (cf.~\Cref{sec:algorithms}) is applicable to AHLTL; even for instances where no approximate solutions were previously known.
In \Cref{sec:implementation}, we show that our prototype model checker for \fphyperltl{} can verify asynchronous properties in practice.

\section{Model-Checking \oldfphyperltl}\label{sec:algorithms}

In general, finite-state model checking of \fphyperltl{} is highly undecidable (cf.~\Cref{corr:sohyper_hardness}). 
In this section, we outline a partial algorithm that computes approximations on the concrete values of second-order variables for a fragment of \fphyperltl{}.
At a very high-level, our algorithm (\Cref{alg:mainVerification}) iteratively computes under- and overapproximations for second-order variables.
It then turns to resolve first-order quantification, using techniques from HyperLTL model checking~\cite{FinkbeinerRS15,BeutnerF23}, and resolves existential and universal trace quantification on the under- and overapproximation of the second-order variables, respectively.
If the verification fails, it goes back to refine second-order approximations. 

In this section, we focus on the setting where we are interested in the least sets (using $\smallestSet$), and use techniques to approximate the \emph{least} fixpoint. 
A similar (dual) treatment is possible for \fphyperltl{} formulas that use the largest set.
Every \fphyperltl{} which uses only minimal sets has the following form:
\begin{align}\label{eq:fpformula}
    \varphi = \gamma_1. (Y_1, \smallestSet, \varphi^\mathit{con}_1). \gamma_2 \ldots.  (Y_k, \smallestSet{}, \varphi^\mathit{con}_k)\ldot \gamma_{k+1}\ldot \psi
\end{align}
We quantify second-order variables $Y_1, \ldots, Y_k$, where, for each $j \in [k]$, $Y_j$ is the least set that satisfies $\varphi^\mathit{con}_j$. 
Finally, for each $j \in [k+1]$, 
\begin{align*}
    \gamma_j = \quant_{l_j+1} \pi_{l_j+1} \in X_{l_j+1} \ldots \quant_{l_{j+1}} \pi_{l_{j+1}} \in X_{l_{j+1}}
\end{align*}
is the block of first-order quantifiers that sits between the quantification of $Y_{j-1}$ and $Y_j$.
Here $X_{l_j+1}, \ldots, X_{l_{j+1}} \in \{\systemvar, \allvar, Y_1, \ldots, Y_{j-1}\}$ are second-order variables that are quantified before $\gamma_j$.
In particular, $\pi_1, \ldots, \pi_{l_j}$ are the first-order variables quantified before~$Y_j$.

\subsection{Fixpoints in \oldfphyperltl{}}

We consider a fragment of \fphyperltl{} which we call the \emph{least fixpoint fragment}.
Within this fragment, we restrict the formulas $\varphi^\mathit{con}_1, \ldots, \varphi^\mathit{con}_k$ such that $Y_1, \ldots, Y_k$ can be approximated as (least) fixpoints. 
Concretely, we say that $\varphi$ is in the \emph{least fixpoint fragment} of \fphyperltl{} if for all $j \in [k]$, $\varphi^\mathit{con}_j$ is a conjunction of formulas of the form 
\begin{align}\label{eq:fixpoint-condition}
    \forall \dot{\pi}_1 \in X_1. \ldots \forall \dot{\pi}_n \in X_n\ldot \psi_{\mathit{step}} \rightarrow \dot{\pi}_M \modin Y_j
\end{align}
where each $X_i \in \{\systemvar, \allvar, Y_1, \ldots, Y_j\}$, $\psi_{\mathit{step}}$ is quantifier-free formula over trace variables $\dot{\pi}_1, \ldots, \dot{\pi}_n, \pi_1, \ldots, \pi_{l_j}$, and $M \in [n]$.
Intuitively, \Cref{eq:fixpoint-condition} states a requirement on traces that should be included in $Y_j$. 
If we find traces $\dot{t}_1 \in X_1, \ldots, \dot{t}_n \in X_n$ that, together with the traces $t_1, \ldots, t_{l_j}$ quantified before $Y_j$, satisfy $\psi_{\mathit{step}}$, then $\dot{t}_M$ should be included in $Y_j$.

Together with the minimality constraint on $Y_j$ (stemming from the semantics of \fphyperltl{}), this effectively defines a (monotone) least fixpoint computation, as $\psi_{\mathit{step}}$ defines exactly the traces to be added to the set. 
This will allow us to use results from fixpoint theory to compute approximations for the sets $Y_j$.

Our least fixpoint fragment captures most properties of interest, in particular, common knowledge (\Cref{sec:running:example}) and asynchronous hyperproperties (\Cref{sec:asynchronous:hyperproperties}).
We observe that formulas of the above form ensure that the solution $Y_j$ is unique, i.e., for any trace assignment $\Pi$ to $\pi_1, \ldots, \pi_{l_j}$ and second-order assignment $\Delta$ to $\systemvar, \allvar, Y_1, \ldots, Y_{j-1}$, there is only one element in $\mathit{sol}(\Pi, \Delta, (Y_j, \smallestSet, \varphi_j^\mathit{con}))$.

\subsection{Functions as Automata}\label{sec:functions:as:automata}

In our (approximate) model-checking algorithm, we represent a concrete assignment to the second-order variables $Y_1, \ldots, Y_k$ using automata $\calB_{Y_1}, \ldots, \calB_{Y_k}$.
The concrete assignment of $Y_j$ can depend on traces assigned to $\pi_1, \ldots, \pi_{l_j}$, i.e., the first-order variables quantified before $Y_j$.
To capture these dependencies, we view each $Y_j$ not as a set of traces but as a function mapping traces of all preceding first-order variables to a set of traces. 
We represent such a function $f : (\Sigma^\omega)^{l_j} \to 2^{(\Sigma^\omega)}$ mapping the $l_j$ traces to a set of traces as an automaton $\calA$ over $\Sigma^{l_j+1}$.
For traces $t_1, \ldots, t_{l_j}$, the set $f(t_1, \ldots, t_{l_j})$ is represented in the automaton by the set $\{ t \in \Sigma^\omega \mid \mathit{zip}(t_1, \ldots, t_{l_j}, t) \in \calL(\calA) \}$.
For example, the function $f(t_1) := \{t_1\}$ can be defined by the automaton that accepts the zipping of a pair of traces exactly if both traces agree on all propositions. 
This representation of functions as automata allows us to maintain an assignment to $Y_j$ that is parametric in $\pi_1, \ldots, \pi_{l_j}$ and still allows first-order model checking on $Y_1, \ldots, Y_k$.

\subsection{Model Checking for First-Order Quantification}\label{sec:first:order:checking}

First, we focus on first-order quantification, and assume that we are given a concrete assignment for each second-order variable as fixed automata $\calB_{Y_1}, \ldots, \calB_{Y_k}$ (where $\calB_{Y_j}$ is an automaton over $\Sigma^{l_j+1}$). Our construction for resolving first-order quantification is based on HyperLTL model checking~\cite{FinkbeinerRS15}, but needs to work on sets of traces that, themselves, are based on traces quantified before (cf.~\Cref{sec:functions:as:automata}).
Recall that the first-order quantifier prefix is $\gamma_1 \cdots \gamma_{k+1} = \quant_{1} \pi_{1} \in X_{1} \cdots \quant_{l_{k+1}} \pi_{l_{k+1}} \in X_{l_{k+1}}$.
For each $1 \leq i \leq l_{k+1}$ we inductively construct an automaton $\calA_i$ over $\Sigma^{i-1}$ that summarizes all trace assignments to $\pi_1, \ldots, \pi_{i-1}$  that satisfy the subformula starting with the quantification of $\pi_i$.
That is, for all traces $t_1, \ldots, t_{i-1}$ we have 
\begin{align*}
    [\pi_1 \mapsto t_1, \ldots, \pi_{i-1} \mapsto t_{i-1}] \models \quant_{i} \pi_{i} \in X_{i} \cdots \quant_{l_{k+1}} \pi_{l_{k+1}} \in X_{l_{k+1}}\ldot \psi
\end{align*}
(under the fixed second-order assignment for $Y_1, \ldots, Y_k$ given by $\calB_{Y_1}, \ldots, \calB_{Y_k}$) if and only if $\mathit{zip}(t_1, \ldots, t_{i-1}) \in \calL(\calA_i)$.
In the context of HyperLTL model checking we say $\calA_i$ is \emph{equivalent} to $\quant_{i} \pi_{i} \in X_{i} \cdots \quant_{l_{k+1}} \pi_{l_{k+1}} \in X_{l_{k+1}}\ldot \psi$ \cite{FinkbeinerRS15,BeutnerF23}. 
In particular, $\calA_1$ is an automaton over singleton alphabet $\Sigma^0$.

We construct $\calA_1, \ldots, \calA_{l_{k+1}+1}$ inductively, starting with $\calA_{l_{k+1}+1}$. 
Initially, we construct $\calA_{l_{k+1}+1}$ (over $\Sigma^{l_{k+1}}$) using a standard LTL-to-NBA construction on the (quantifier-free) body $\psi$ (see \cite{FinkbeinerRS15} for details).
Now assume that we are given an (inductively constructed) automaton $\calA_{i+1}$ over $\Sigma^i$ and want to construct~$\calA_i$. 
We first consider the case where $\quant_i = \exists$, i.e., the $i$th trace quantification is existential.
Now $X_i$ (the set where $\pi_i$ is resolved on) either equals $\systemvar$, $\allvar$ or $Y_j$ for some $j \in [k]$.
In either case, we represent the current assignment to $X_i$ as an automaton $\calC$ over $\Sigma^{T + 1}$ for some $T < i$ that defines the model of $X_i$ based on traces $\pi_1, \ldots, \pi_T$:
In case $X_i = \systemvar$, we set $\calC$ to be the automaton over $\Sigma^{0 + 1}$ that accepts exactly the traces in the given system $\calT$; in case $X_i = \allvar$, we set $\calC$ to be the automaton over $\Sigma^{0 + 1}$ that accepts all traces; If $X_i = Y_j$ for some $j \in [k]$ we set $\calC$ to be $\calB_{Y_j}$ (which is an automaton over $\Sigma^{l_j + 1}$).\footnote{Note that in this case $l_j < i$: if trace $\pi_i$ is resolved on $Y_j$ (i.e,  $X_i = Y_j$), then $Y_j$ must be quantified \emph{before} $\pi_i$ so there are at most $i-1$ traces quantified before~$Y_j$. }
Given $\calC$, we can now modify the construction from \cite{FinkbeinerRS15}, to resolve first-order quantification:
The desired automaton~$\calA_i$ should accept the zipping of traces $t_1, \ldots, t_{i-1}$ if there exists a trace $t$ such that (1) $\mathit{zip}(t_1, \ldots, t_{i-1}, t) \in \calL(\calA_{i+1})$, \emph{and} (2) the trace $t$ is contained in the set of traces assigned to $X_i$ as given by $\calC$, i.e., $\mathit{zip}(t_1, \ldots, t_T, t) \in \calL(\calC)$.
The construction of this automaton is straightforward by taking a product of~$\calA_{i+1}$ and~$\calC$.
We denote this automaton with  \lstinline[style=default, language=custom-lang]|eProduct($\calA_{i+1}$,$\calC$)|.
In case $\quant_i = \forall$ we exploit the duality that $\forall \pi. \psi$ = $\neg \exists \pi. \neg \psi$, combining the above construction with automata complementation. We denote this universal product of $\calA_{i+1}$ and $\calC$ with \lstinline[style=default, language=custom-lang]|uProduct($\calA_{i+1}$,$\calC$)|.

The final automaton $\calA_1$ is an automaton over singleton alphabet $\Sigma^0$ that is equivalent to $\gamma_1 \cdots \gamma_{k+1}. \psi$, i.e., the entire first-order quantifier prefix.
Automaton $\calA_1$ thus satisfies $\calL(\calA_1) \neq \emptyset$ (which we can decide) iff the empty trace assignment satisfies the first-order formula $\gamma_1 \cdots \gamma_{k+1}\ldot \psi$, iff $\varphi$ (of \Cref{eq:fpformula}) holds within the fixed model for $Y_1, \ldots, Y_k$.
For a given fixed second-order assignment (given as automata $\calB_{Y_1}, \ldots, \calB_{Y_k}$), we can thus decide if the system satisfies the first-order part.

During the first-order model-checking phase, each quantifier alternations in the formula require complex automata complementation. 
For the first-order phase, we could also use cheaper approximate methods by, e.g., instantiating the existential trace using a strategy \cite{CoenenFST19,BeutnerF22CAV,BeutnerF22}.

\subsection{Bidirectional Model Checking}

So far, we have discussed the verification of the first-order quantifiers assuming we have a fixed model for all second-order variables $Y_1, \ldots, Y_k$.
In our actual model-checking algorithm, we instead maintain under- and overapproximations on each of the $Y_1, \ldots, Y_k$.

\begin{algorithm}[!t]
\caption{}\label{alg:mainVerification}
\begin{code}
verify($\varphi$, $\calT$) = 
@@let $\varphi$ = $\big[\gamma_j \; (Y_j, \smallestSet , \varphi_j^\mathit{con})\big]_{j=1}^k \; \gamma_{k+1}\ldot \psi$ where $\gamma_i$ = $\big[\quant_m \pi_m \in X_m \big]_{m=l_i+1}^{l_{i+1}} $
@@let N = 0
@@let $\calA_\calT$ = systemToNBA($\calT$)
@@repeat 
@@@@// Start outside-in traversal on second-order variables
@@@@let $\flat$ = $\big[\systemvar \mapsto (\calA_\calT,\calA_\calT), \allvar \mapsto (\calA_\top, \calA_\top)\big]$
@@@@for $j$ from $1$ to $k$ do 
@@@@@@$\calB_j^l$ := underApprox($(Y_j,\smallestSet, \varphi_j^\mathit{con})$,$\flat$,$N$)
@@@@@@$\calB_j^u$ := overApprox($(Y_j,\smallestSet, \varphi_j^\mathit{con})$,$\flat$,$N$)
@@@@@@$\flat(Y_j)$ := $(\calB_j^l,\calB_j^u)$
@@@@// Start inside-out traversal on first-order variables
@@@@let $\calA_{l_{k+1} + 1}$ = LTLtoNBA($\psi$)
@@@@for $i$ from $l_{k+1}$ to $1$ do 
@@@@@@let $(\calC^l, \calC^u)$ = $\flat(X_i)$
@@@@@@if $\quant_i = \exists$ then 
@@@@@@@@$\calA_i$ := eProduct($\calA_{i+1}$, $\calC^l$) 
@@@@@@else 
@@@@@@@@$\calA_i$ := uProduct($\calA_{i+1}$, $\calC^u$)
@@@@if $\calL(\calA_1) \neq \emptyset$ then 
@@@@@@return SAT 
@@@else 
@@@@@@$N$ = $N$ + 1
\end{code}
\end{algorithm}

In each iteration, we first traverse the second-order quantifiers in an \emph{outside-in} direction and compute lower- and upper-bounds on each $Y_j$.
Given the bounds, we then traverse the first-order prefix in an \emph{inside-out} direction using the current approximations to $Y_1, \ldots, Y_k$.
If the current approximations are not precise enough to witness the satisfaction (or violation) of a property, we repeat and try to compute better bounds on $Y_1, \ldots, Y_k$. 
Due to the different directions of traversal, we refer to our model-checking approach as \emph{bidirectional}.
\Cref{alg:mainVerification} provides an overview.
Initially, we convert the system $\calT$ to an NBA $\calA_\calT$ accepting exactly the traces of the system.
In each round, we compute under- and overapproximations for each $Y_j$ in a mapping~$\flat$.
We initialize $\flat$ by mapping $\systemvar$ to $(\calA_\calT, \calA_\calT)$ (i.e., the value assigned to the system variable is precisely~$\calA_\calT$ for both under- and overapproximation), and $\allvar$ to $(\calA_\top, \calA_\top)$ where $\calA_\top$ is an automaton over $\Sigma^1$ accepting all traces. 
We then traverse the second-order quantifiers outside-in (from $Y_1$ to $Y_k$) and for each $Y_j$ compute a pair $(\calB_j^l, \calB_j^u)$ of automata over $\Sigma^{l_j + 1}$ that under- and overapproximate the actual (unique) model of $Y_j$.  
We compute these approximations using functions \lstinline[style=default, language=custom-lang]|underApprox| and \lstinline[style=default, language=custom-lang]|overApprox|, which can be instantiated with any procedure that computes sound lower and upper bounds (see \Cref{sec:computing:approximations}). 
During verification, we further maintain a precision bound~$N$ (initially set to $0$) that tracks the current precision of the second-order approximations.

When $\flat$ contains an under- and overapproximation for each second-order variable, we traverse the first-order variables in an inside-out direction (from~$\pi_{l_{k+1}}$ to~$\pi_1$) and, following the construction outlined in \Cref{sec:first:order:checking}, construct automata $\calA_{l_k+1}, \ldots, \calA_1$.
Different from the simplified setting in \Cref{sec:first:order:checking} (where we assume a fixed automaton $\calB_{Y_j}$ providing a model for each $Y_j$), the mapping $\flat$ contains only approximations of the concrete solution. We choose which approximation to use according to the corresponding set quantification:
In case we construct $\calA_i$ and $\quant_i = \exists$, we use the \emph{underapproximation} (thus making sure that any witness trace we pick is indeed contained in the actual model of the second-order variable); and if $\quant_i = \forall$, we use the \emph{overapproximation} (making sure that we consider at least those traces that are in the actual solution). 
If $\calL(\calA_1)$ is non-empty, i.e., accepts the empty trace assignment, the formula holds (assuming the approximations returned by  \lstinline[style=default, language=custom-lang]|underApprox| and \lstinline[style=default, language=custom-lang]|overApprox| are sound). 
If not, we increase the precision bound $N$ and repeat.

In \Cref{alg:mainVerification}, we only check for the satisfaction of a formula (to keep the notation succinct). 
Using the second-order approximations in $\flat$ we can also check the negation of a formula (by considering the negated body and dualizing all trace quantifiers). 
Our tool (\Cref{sec:implementation}) makes use of this and thus simultaneously tries to show satisfaction and violation of a formula.

\subsection{Computing Under- and Overapproximations}\label{sec:computing:approximations}

In this section we provide concrete instantiations for \lstinline[style=default, language=custom-lang]|underApprox| and \lstinline[style=default, language=custom-lang]|overApprox|.

\subsubsection*{Computing Underapproximations.} \label{sec:computing:under:approximations}
As we consider the fixpoint fragment, each formula $\varphi_j^\mathit{con}$ (defining $Y_j$) is a conjunction of formulas of the form in \Cref{eq:fixpoint-condition}, thus defining $Y_j$ via a least fixpoint computation.
For simplicity, we assume that $Y_j$ is defined by the single conjunct, given by~\Cref{eq:fixpoint-condition} (our construction generalizes easily to a conjunction of such formulas). 
Assuming fixed models for $\systemvar$, $\allvar$ and $Y_1, \ldots, Y_{j-1}$, the fixpoint operation defining $Y_j$ is monotone, i.e., the larger the current model for $Y_j$ is, the more traces we need to add according to \Cref{eq:fixpoint-condition}.
Monotonicity allows us to apply the Knaster–Tarski theorem \cite{tarski1955lattice} and compute underapproximations to the fixpoint by iteration.

In our construction of an approximation for $Y_j$, we are given a mapping $\flat$ that fixes a pair of automata for $\systemvar$, $\allvar$, and $Y_1, \ldots, Y_{j-1}$ (due to the outside-in traversal in \Cref{alg:mainVerification}). 
As we are computing an underapproximation, we use the underapproximation for each of the second-order variables in $\flat$.
So $\flat(\systemvar)$ and $\flat(\allvar)$ are automata over $\Sigma^1$ and for each $j' \in [j-1]$, $\flat(Y_{j'})$ is an automaton over $\Sigma^{l_{j'} + 1}$.
Given this fixed mapping $\flat$, we iteratively construct automata $\hat{\calC}_0, \hat{\calC}_1, \ldots$ over $\Sigma^{l_j + 1}$ that capture (increasingly precise) underapproximations on the solution for $Y_j$.
We set~$\hat{\calC}_0$ to be the automaton with the empty language.
We then recursively define $\hat{\calC}_{N+1}$ based on $\hat{\calC}_{N}$ as follows:
For each second-order variable $X_i$ for $i \in [n]$ used in \Cref{eq:fixpoint-condition} we can assume a concrete assignment in the form of an automaton $\calD_i$ over $\Sigma^{T_i + 1}$ for some $T_i \leq l_j$:
In case $X_i \neq Y_j$ (so $X_i \in \{\systemvar, \allvar, Y_1, \ldots, Y_{j-1}\}$), we set $\calD_i := \flat(X_i)$.
In case $X_i = Y_j$, we set $\calD_i := \hat{\calC}_N$, i.e., we use the current approximation of $Y_j$ in iteration $N$.
After we have set $\calD_1, \ldots, \calD_n$, we compute an automaton $\dot{\calC}$ over $\Sigma^{l_j+1}$ that accepts $\mathit{zip}(t_1, \ldots, t_{l_j}, t)$
iff there exists traces $\dot{t}_1, \ldots, \dot{t}_n$ such that (1) $\mathit{zip}(t_1, \ldots, t_{T_i}, \dot{t}_i) \in \calL(\calD_i)$ for all $i \in [n]$, (2) $[\pi_1 \mapsto t_1, \ldots, \pi_{l_j} \mapsto t_{l_j}, \dot\pi_1 \mapsto \dot{t}_{1},\ldots,  \dot\pi_n \mapsto \dot{t}_n] \models \psi_\mathit{step}$, and (3) trace $t$ equals $\dot{t}_M$ (of~\Cref{eq:fixpoint-condition}).
The intuition is that $\dot{\calC}$ captures all traces that should be added to $Y_j$:
Given $t_1, \ldots, t_{l_j}$ we check if there are traces $\dot{t}_1, \ldots, \dot{t}_n$ for trace variables $\dot{\pi}_1, \ldots, \dot{\pi}_n$ in \Cref{eq:fixpoint-condition} where (1) each $\dot{t}_i$ is in the assignment for $X_i$, which is captured by the automaton $\calD_i$ over $\Sigma^{T_i + 1}$, and (2) the traces $\dot{t}_1, \ldots, \dot{t}_n$ satisfy $\varphi_\mathit{step}$. If this is the case, we want to add $\dot{t}_M$ (as stated in \Cref{eq:fixpoint-condition}). 
We then define $\hat{\calC}_{N+1}$ as the union of $\hat{\calC}_{N}$ and $\dot{\calC}$, i.e. extend the previous model with all (potentially new) traces that need to be added.

\subsubsection*{Computing Overapproximations.}\label{sec:computing:over:approximations}
As we noted above, conditions of the form of \Cref{eq:fixpoint-condition} always define fixpoint constraints. 
To compute upper bounds on such fixpoint constructions we make use of Park's theorem, \cite{Winskel93} stating that if we find some set (or automaton) $\calB$ that is inductive (i.e., when computing all traces that we would need to add assuming the current model of $Y_j$ is $\calB$, we end up with traces that are already in $\calB$), then $\calB$ overapproximates the unique solution (aka.~least fixpoint) of $Y_j$.
To derive such an inductive invariant, we employ techniques developed in the context of regular model checking \cite{BouajjaniJNT00} (see \Cref{sec:relatedWork}). 
Concretely, we employ the approach from \cite{ChenHLR17} that uses automata learning \cite{DBLP:journals/iandc/Angluin87} to find suitable invariants. 
While the approach from \cite{ChenHLR17} is limited to finite words, we extend it to an $\omega$-setting by interpreting an automaton accepting finite words as one that accepts an $\omega$-word $u$ iff every prefix of $u$ is accepted.\footnote{This effectively poses the assumption that the step formula specifies a safety property, which seems to be the case for almost all examples.
As an example, common knowledge infers a safety property:
In each step, we add all traces for which there exists some trace that agrees on all propositions observed by that agent.}
As soon as the learner provides a candidate for an equivalence check, we check that it is inductive and, if not, provide some finite counterexample (see \cite{ChenHLR17} for details).
If the automaton is inductive, we return it as a potential overapproximation. 
Should this approximation not be precise enough, the first-order model checking (\Cref{sec:first:order:checking}) returns some concrete counterexample, i.e., some trace contained in the invariant but violating the property, which we use to provide more counterexamples to the learner.

\section{Implementation and Experiments}\label{sec:implementation}

We have implemented our model-checking algorithm in a prototype tool we call \texttt{HySO} (\textbf{Hy}perproperties with \textbf{S}econd \textbf{O}rder).\footnote{Our tool is publicly available at \url{https://zenodo.org/record/7878311}.}
Our tool uses \texttt{spot} \cite{DBLP:conf/cav/Duret-LutzRCRAS22} for basic automata operations (such as LTL-to-NBA translations and complementations). 
To compute under- and overapproximations, we use the techniques described in \Cref{sec:computing:approximations}.
We evaluate the algorithm on the following benchmarks.

\subsubsection*{Muddy Children.}
The muddy children puzzle \cite{DBLP:books/mit/FHMV1995} is one of the classic examples in common knowledge literature.
The puzzle consists of $n$ children standing such that each child can see all other children's faces.
From the $n$ children, an unknown number $k \geq 1$ have a muddy forehead, and in incremental rounds, the children should step forward if they know if their face is muddy or not.
Consider the scenario of $n=2$ and $k = 1$, so child $a$ sees that child $b$ has a muddy forehead and child $b$ sees that $a$ is clean.
In this case, $b$ immediately steps forward, as it knows that its forehead is muddy since 
$k\geq 1$.
In the next step, $a$ knows that its face is clean since $b$ stepped forward in round 1.
In general, one can prove that all children step forward in round $k$, deriving common knowledge. 

For each $n$ we construct a transition system $\calT_n$ that encodes the muddy children scenario with $n$ children.
For every $m$ we design a \fphyperltl{} formula $\varphi_m$ that adds to the common knowledge set $X$ all traces that appear indistinguishable in the first $m$ steps for some child.
We then specify that all traces in $X$ should agree on all inputs, asserting that all inputs are common knowledge.\footnote{This property is not expressible in non-hyper logics such as LTL$_{\calK, \sfC}$, where we can only check \emph{trace properties} on the common knowledge set $X$. In contrast, \fphyperltl{} allows us to check \emph{hyperproperties} on $X$. That way, we can express that some value is common knowledge (i.e., equal across all traces in the set) and not only that a property is common knowledge (i.e., holds on all traces in the set).}
We used \texttt{HySO} to \emph{fully automatically} check $\calT_n$ against $\varphi_m$ for varying values of~$n$ and $m$, i.e., we checked if, after the first $m$ steps, the inputs of all children are common knowledge. 
As expected, the above property holds only if $m \geq n$ (in the worst case, where all children are dirty $(k = n)$, the inputs of all children only become common knowledge after $n$ steps). 
We depict the results in \Cref{tab:muddy}.

\begin{table}[!t]
\begin{subtable}[b]{0.45\textwidth}
    \begin{center}
    \def\arraystretch{1.3}
    \setlength\tabcolsep{2mm}
    \begin{tabular}{cccccc}
        & & \multicolumn{4}{c}{\textbf{m}} \\
        \cline{3-6}
        & & \multicolumn{1}{|c|}{1} & \multicolumn{1}{|c|}{2} & \multicolumn{1}{|c|}{3} & \multicolumn{1}{|c|}{4} \\
        \cline{2-6}
        \multirow{6}{*}{\textbf{n}} & \multicolumn{1}{|c|}{2} & \multicolumn{1}{|c|}{\thead{\xmark{}\\0.64}} & \multicolumn{1}{|c|}{\thead{\cmark{}\\ 0.59}} &\multicolumn{1}{|c|}{} &  \multicolumn{1}{|c|}{}\\
        \cline{2-6}
        & \multicolumn{1}{|c|}{3} & \multicolumn{1}{|c|}{\thead{\xmark{}\\ 0.79}} & \multicolumn{1}{|c|}{\thead{\xmark{}\\ 0.75}} &\multicolumn{1}{|c|}{\thead{\cmark{}\\0.54}} &  \multicolumn{1}{|c|}{}\\
        \cline{2-6}
        & \multicolumn{1}{|c|}{4} & \multicolumn{1}{|c|}{\thead{\xmark{}\\2.72}} & \multicolumn{1}{|c|}{\thead{\xmark{}\\ 2.21}} &\multicolumn{1}{|c|}{\thead{\xmark{}\\ 1.67}} &  \multicolumn{1}{|c|}{\thead{\cmark{}\\1.19}}\\
        \cline{2-6}
    \end{tabular}
    \end{center}
    \subcaption{}\label{tab:muddy}

\end{subtable}%
\hfil
\begin{subtable}[b]{0.55\textwidth}
    \begin{center}
    \def\arraystretch{1.5}
    \setlength\tabcolsep{2mm}
    \begin{tabular}{llll}
        \toprule
        \textbf{Instance} & \textbf{Method} & \textbf{Res} & $\boldsymbol{t}$  \\
        \midrule
        $\calT_\mathit{syn}$, $\varphi_\mathit{OD}$ & - & \cmark{} & 0.26\\
        $\calT_\mathit{asyn}$, $\varphi_\mathit{OD}$ & - & \xmark{} & 0.31\\
        $\calT_\mathit{syn}$, $\varphi^\mathit{asyn}_\mathit{OD}$ & Iter (0) & \cmark{} & 0.50\\
        $\calT_\mathit{syn}$, $\varphi^\mathit{asyn}_\mathit{OD}$ & Iter (1) & \cmark{} & 0.78\\
        \textsc{Q1}, $\varphi_\mathit{OD}$ & - & \xmark{} & 0.34\\
        \textsc{Q1}, $\varphi^\mathit{asyn}_\mathit{OD}$ & Iter (1) & \cmark{} & 0.86\\
        \bottomrule  
    \end{tabular}
    \end{center}
    \subcaption{
    }\label{tab:more}
\end{subtable}%

\caption{
In \Cref{tab:muddy}, we check common knowledge in the muddy children puzzle for $n$ children and $m$ rounds. We give the result (\cmark{} if common knowledge holds and \xmark{} if it does not), and the running time.
In \Cref{tab:more}, we check synchronous and asynchronous versions of observational determinism. We depict the number of iterations needed and running time.
Times are given in seconds.
}
\end{table}

\subsubsection*{Asynchronous Hyperproperties.}
As we have shown in \Cref{sec:asynchronous:hyperproperties}, we can encode arbitrary AHLTL properties into \fphyperltl{}.
We verified synchronous and asynchronous version of observational determinism (cf.~\Cref{sec:asynchronous:hyperproperties}) on programs taken from \cite{BaumeisterCBFS21,BeutnerF21,BeutnerF23LMCS}.
We depict the verification results in \Cref{tab:more}.
Recall that \fphyperltl{} properties without any second-order variables correspond to HyperQPTL formulas.
\texttt{HySO} can check such properties precisely, i.e., it constitutes a sound-and-complete model checker for HyperQPTL properties with an arbitrary quantifier prefix.
The synchronous version of observational determinism is a HyperLTL property and thus needs no second-order approximation (we set the method column to ``-'' in these cases).

\subsubsection*{Common Knowledge in Multi-agent Systems.}
We used \texttt{HySO} for an automatic analysis of the system in \Cref{fig:common:knowledge}. 
Here, we verify that on initial trace $\{a\}^n\{d\}^\omega$ it is CK that $a$ holds in the first step.
We use a similar formula as the one of \Cref{sec:running:example}, with the change that we are interested in whether $a$ is CK (whereas we used $\LTLnext a$ in \Cref{sec:running:example}).
As expected, \texttt{HySO} requires $2n-1$ iterations to converge. 
We depict the results in \Cref{tab:running}.

\begin{table}[!t]
\begin{subtable}[b]{0.4\textwidth}
    \begin{center}
        \def\arraystretch{1.5}
        \setlength\tabcolsep{1.7mm}
        \begin{tabular}{llll}
            \toprule
            $\boldsymbol{n}$ & \textbf{Method} & \textbf{Res} & $\boldsymbol{t}$  \\
            \midrule
            1 & Iter (1) & \cmark{} & 0.51\\
            2 & Iter (3) & \cmark{}& 0.83\\
            3 & Iter (5) & \cmark{}& 1.20\\
            10 & Iter (19) & \cmark{}& 3.81\\
            100 & Iter (199) & \cmark& 102.8\\
            \bottomrule
        \end{tabular}
    \end{center}
    \subcaption{
    }\label{tab:running}

\end{subtable}%
\hfill
\begin{subtable}[b]{0.6\textwidth}
\begin{center}
    \def\arraystretch{1.5}
    \setlength\tabcolsep{1.7mm}
    \begin{tabular}{llll}
        \toprule
        \textbf{Instance} & \textbf{Method} & \textbf{Res} & $\boldsymbol{t}$  \\
        \midrule
        \textsc{SwapA} & Learn & \cmark{} & 1.07\\
        \textsc{SwapATwice} & Learn & \cmark{} & 2.13\\
        \textsc{SwapA}$_5$ & Iter (5) & \cmark{} & 1.15\\
        \textsc{SwapA}$_{15}$ & Iter (15) & \cmark{} & 3.04\\
        \textsc{SwapAViolation}$_5$ & Iter (5) & \xmark{} & 2.35\\
        \textsc{SwapAViolation}$_{15}$ & Iter (15) & \xmark{} & 4.21\\
        \bottomrule  
    \end{tabular}
\end{center}
\subcaption{}\label{tab:mar}
\end{subtable}%

\caption{
In \Cref{tab:muddy}, we check common knowledge in the example from \Cref{fig:common:knowledge} when starting with $a^nd^\omega$ for varying values of $n$. We depict the number of refinement iterations, the result, and the running time. 
In \Cref{tab:mar}, we verify various properties on Mazurkiewicz traces. We depict whether the property could be verified or refuted by iteration or automata learning, the result, and the time.
Times are given in seconds.}
\end{table}

\subsubsection*{Mazurkiewicz Traces.}
Mazurkiewicz traces are an important concept in the theory of distributed computing \cite{DR1995}.
Let $I \subseteq \Sigma \times \Sigma$ be an independence relation that determines when two consecutive letters can be switched (think of two actions in disjoint processes in a distributed system). 
Any $t \in \Sigma^\omega$ then defines the set of all traces that are equivalent to $t$ by flipping consecutive independent actions an arbitrary number of times (the equivalence class of all these traces is called the Mazurkiewicz Trace). 
See \cite{DR1995} for details. 
The verification problem for Mazurkiewicz traces now asks if, given some $t \in \Sigma^\omega$, all traces in the Mazurkiewicz trace of $t$ satisfy some property $\psi$.
Using \fphyperltl{} we can directly reason about the Mazurkiewicz Trace of any given trace, by requiring that all traces that are equal up to one swap of independent letters are also in a given set (which is easily expressed in \fphyperltl{}).

Using \texttt{HySO} we verify a selection of such trace properties that often require non-trivial reasoning by coming up with a suitable invariant. 
We depict the results in \Cref{tab:mar}. 
In our preliminary experiments, we model a situation where we start with $\{a\}^1\{\}^\omega$ and can swap letters $\{a\}$ and $\{\}$.
We then, e.g., ask if on any trace in the resulting Mazurkiewicz trace, $a$ holds at most once, which requires inductive invariants and cannot be established by iteration.

\section{Related Work}\label{sec:relatedWork}

In recent years, many logics for the formal specification of hyperproperties have been developed, extending temporal logics with explicit path quantification (examples include HyperLTL, HyperCTL$^*$ \cite{ClarksonFKMRS14}, HyperQPTL \cite{DBLP:phd/dnb/Rabe16,BeutnerF23LPAR}, HyperPDL \cite{GutsfeldMO20}, and HyperATL$^*$ \cite{BeutnerF21,BeutnerF23LMCS}); extending first and second-order logics with an equal level predicate \cite{CoenenFST19,Finkbeiner017}; or extending ($\omega$)-regular hyperproperties~\cite{DBLP:conf/vmcai/GoudsmidGS21,DBLP:conf/lata/BonakdarpourS21} to context-free hyperproperties~\cite{DBLP:journals/corr/abs-2209-10306}. 
\sohyperltl{} is the first temporal~logic that reasons about second-order hyperproperties which allows is to capture many existing (epistemic, asynchronous, etc.) hyperlogics while at the same time taking advantage of model-checking solutions that have been proven successful in first-order~settings.

\paragraph{Asynchronous Hyperproperties.}
For asynchronous hyperproperties, Gutfeld et al.~\cite{GutsfeldMO21} present an asynchronous extension of the polyadic $\mu$-calculus. 
Bozelli et al.~\cite{BozzelliPS21} extend HyperLTL with temporal operators that are only evaluated if the truth value of some temporal formula changes. 
Baumeister et al. present AHLTL~\cite{BaumeisterCBFS21}, that extends HyperLTL with a explicit quantification over trajectories and can be directly encoded within \fphyperltl{}.

\paragraph{Regular Model Checking.}
Regular model checking \cite{BouajjaniJNT00} is a general verification method for (possibly infinite state) systems, in which each state of the system is interpreted as a finite word.
The transitions of the system are given as a finite-state (regular) transducer, and the model checking problem asks if, from some initial set of states (given as a regular language), some bad state is eventually reachable. 
Many methods for automated regular model checking have been developed \cite{BoigelotLW03,DamsLS01,BoigelotLW04,ChenHLR17}.
\sohyperltl{} can be seen as a logical foundation for $\omega$-regular model checking:
Assume the set of initial states is given as a QPTL formula $\varphi_\mathit{init}$, the set of bad states is given as a QPTL formula $\varphi_\mathit{bad}$, and the transition relation is given as a QPTL formula  $\varphi_\mathit{step}$ over trace variables $\pi$ and $\pi'$.
The set of bad states is reachable from a trace (state) in $\varphi_\mathit{init}$ iff the following \fphyperltl{} formula holds on the system that generates all traces:
\begin{align*}
&\big(X, \smallestSet, \forall \pi \in \systemvar\ldot \varphi_\mathit{init}(\pi) \rightarrow \pi \modin X \land\\
&\quad\quad\quad\forall \pi \in X. \forall \pi' \in \systemvar\ldot \varphi_\mathit{step}(\pi, \pi') \rightarrow \pi' \modin X\big)\ldot \forall \pi \in X. \neg \varphi_\mathit{bad}(\pi)
\end{align*}

Conversely, \fphyperltl{} can express more complex properties, beyond the reachability checks possible in the framework of ($\omega$-)regular model checking.

\paragraph{Model Checking Knowledge.}
Model checking of knowledge properties in multi-agent systems was developed in the tools \texttt{MCK}~\cite{DBLP:conf/cav/GammieM04} and \texttt{MCMAS}~\cite{DBLP:journals/sttt/LomuscioQR17}, which can exactly express LTL$_{\calK}$.
Bozzelli et al.~\cite{DBLP:conf/fossacs/BozzelliMP15} have shown that HyperCTL$^*$ and LTL$_\calK$ have incomparable expressiveness, and present HyperCTL$^*_{lp}$~ -- an extension of HyperCTL$^*$ that can reason about past -- to unify  HyperCTL$^*$ and LTL$_\calK$.
While HyperCTL$^*_{lp}$ can express the knowledge operator, it cannot capture common knowledge. LTL$_{\calK,\sfC}$~\cite{DBLP:conf/spin/HoekW02} captures both knowledge and common knowledge, but the suggested model-checking algorithm only handles a decidable fragment that is reducible to LTL model checking. 
\section{Conclusion}

Hyperproperties play an increasingly important role in many areas of computer science. There is a strong need for specification languages and verification methods that reason about hyperproperties in a uniform and general manner, similar to what is standard for more traditional notions of safety and reliability.
In this paper, we have ventured forward from the first-order reasoning of logics like HyperLTL into the realm of second-order hyperproperties, i.e., properties that not only compare individual traces but reason comprehensively about \emph{sets} of such traces. 
With \sohyperltl{}, we have introduced a natural specification language and a general model-checking approach for second-order hyperproperties. \sohyperltl{} provides a general framework for a wide range of relevant hyperproperties, including common knowledge and asynchronous hyperproperties, which could previously only be studied with specialized logics and algorithms. \sohyperltl{} also provides a starting point for future work on second-order hyperproperties in areas such as cyber-physical~\cite{10.1145/3127041.3127058} and probabilistic systems~\cite{DimitrovaFT20}.

\subsubsection*{Acknowledgements.} 
We thank Jana Hofmann for the fruitful discussions.  
This work was supported by the European Research Council (ERC) Grant HYPER (No. 101055412), by DFG grant 389792660 as part of TRR~248, and by the German Israeli Foundation (GIF) Grant No. I-1513-407.2019.

\bibliographystyle{splncs04}
\bibliography{bib}
\appendix

\section{Additional Material for \Cref{sec:second:order:hyperltl}}

\subsection{Additional Material for \Cref{sec:sohyperltl}}\label{app:qptl}
\hyperltlIntoSoHyper*
\begin{proof}
HyperQPTL extends HyperLTL with quantification over atomic propositions. However, HyperQPTL can also be expressed using quantification over arbitrary traces, that are not necessarily system traces, capturing exactly the same semantics. Then, we can translate every HyperLTL trace quantification $\quant \pi. \varphi$ to $\quant \pi \in \systemvar. \varphi$, and every HyperQPTL trace quantification $\quant \tau. \varphi$ to $\quant \tau \in \allvar. \varphi$, to obtain a \sohyperltl{} formula. \qed
\end{proof}

\subsection{Additional Material for \Cref{sec:fphyperltl}}\label{app:fptoso}

\fpToSo*
\begin{proof}
    We translate the \fphyperltl{} formula $\varphi$ into a \sohyperltl{} formula $\llbracket \varphi \rrbracket$:
    \begin{align*}
        \llbracket \psi \rrbracket &:= \psi \\
        \llbracket \quant \pi \in X\ldot \varphi \rrbracket &:= \quant \pi \in X\ldot \llbracket \varphi \rrbracket \\
        \llbracket \exists  (X, \smallestSet, \varphi_1)\ldot \varphi_2 \rrbracket &:= \exists X\ldot \llbracket \varphi_1 \rrbracket \land \big(\forall Y\ldot Y \subsetneq X \Rightarrow \neg \llbracket\varphi_1[Y / X] \rrbracket\big) \land \llbracket \varphi_2 \rrbracket\\
        \llbracket \forall  (X, \smallestSet, \varphi_1)\ldot \varphi_2 \rrbracket &:= \forall X\ldot \Big( \llbracket \varphi_1 \rrbracket \land \big(\forall Y\ldot Y \subsetneq X \Rightarrow \neg \llbracket\varphi_1[Y / X] \rrbracket\big)\Big) \Rightarrow \llbracket \varphi_2 \rrbracket\\
        \llbracket \exists  (X, \largestSet, \varphi_1)\ldot \varphi_2 \rrbracket &:= \exists X\ldot \llbracket \varphi_1 \rrbracket \land \big(\forall Y\ldot Y \supsetneq X\Rightarrow \neg \llbracket\varphi_1[Y / X] \rrbracket\big) \land \llbracket \varphi_2 \rrbracket\\
        \llbracket \forall  (X, \largestSet, \varphi_1)\ldot \varphi_2 \rrbracket &:= \forall X\ldot \Big(\llbracket \varphi_1 \rrbracket \land \big(\forall Y\ldot Y \supsetneq X \Rightarrow \neg \llbracket\varphi_1[Y / X] \rrbracket\big) \Big) \Rightarrow \llbracket \varphi_2 \rrbracket
    \end{align*}
    Path formulas and first-order quantification can be translated verbatim. 
    To translate (fixpoint-based) second-order quantification we use additional second order quantification to express the fact that a set should be a least or greatest.
    We write $Y \subsetneq X$ as a shorthand for 
    \begin{align*}
        \big(\forall \pi \in Y. \pi \modin X \big) \land \big(\exists \pi \in Y. \forall \pi' \in Y\ldot \F \neg (\pi =_\ap \pi' ) \big)
    \end{align*}
     $\varphi_1[Y / X]$ denotes the formula where all free occurrences of $X$ are replaced by $Y$.

    Note that above formula is no \sohyperltl{} formula as it is not in prenex normal form. However, no (first or second-order) quantification occurs under temporal operators, so we can easiliy bring it into prenex normal form.
    It is easy to see that any system satisfies $\varphi$ in the \fphyperltl{} semantics iff it satisfies $\llbracket \varphi \rrbracket$ in the \sohyperltl{} semantics. \qed
\end{proof}

\subsection{Additional Proofs for \Cref{sec:sohyperltl_mc}}\label{app:proofs}

\fphyperUndec*
\begin{proof}
Instead of working with Turing machines, we consider two counter machines as they are simpler to handle. 
A two-counter machine (2CM) maintains two counters $c_1, c_2$ and has a finite set of instructions $l_1, \ldots, l_n$.
Each instruction $l_i$ is of one of the following forms, where $x \in \set{1,2}$ and  $1 \leq j,k \leq n$.
\begin{itemize}
\item $l_i : \big[c_x \coloneqq c_x+1; \texttt{ goto } \{l_j, l_{k}\}\big]$ 
\item $l_i : \big[c_x \coloneqq c_x-1; \texttt{ goto } \{l_j, l_{k}\}\big]$ 
\item $l_i : \big[\texttt{if } c_x = 0 \texttt{ then goto } l_j \texttt{ else goto } l_k\big]$
\end{itemize}
Here, \texttt{goto} $\{l_j, l_{k}\}$ indicates that the machine nondeterministically chooses between instructions $l_j$ and $l_k$.
A configuration of a 2CM is a tuple $(l_i, v_1, v_2)$, where $l_i$ is the next instruction to be executed, and $v_1, v_2 \in \nat$ denote the values of the counters.
The initial configuration of a 2CM is $(l_1, 0, 0)$. 
The transition relation between configurations is defined as expected. 
Decrementing a counter that is already $0$ leaves the counter unchanged. 
An infinite computation is \emph{recurring}  if it visits instruction $l_1$ infinitely many times.
Deciding if a machine has a recurring computation is $\Sigma_1^1$-hard \cite{AlurH94}. 

For our proof, we encode each configuration of the 2CM as an infinite trace. 
We use to atomic propositions $c_1, c_2$, which we ensure to hold exactly once and use this unique position to represent the counter value. 
The current instruction can be encoded in the first position using APs $l_1, \ldots, l_n$. 
We further use a fresh AP $\dagger$ -- which also holds exactly once -- to mark the step of this configuration.
We are interested in a set of traces $X$ that satisfies all of the following requirements:
\begin{enumerate}
    \item The set $X$ contains the initial configuration. Note that in this configuration, $\dagger$ holds in the \emph{first} step.
    \item For every configuration, there exists a successor configuration. Note that in the successor configuration, $\dagger$ is shifted by one position. 
    \item\label{item:uni} All pairs of traces where $\dagger$ holds at the same position are equal. $X$ thus assigns a unique configuration to each step.
    \item\label{item:cons} The computation is recurrent. As already done in \cite{BeutnerCFHK22}, we can ensure this by adding a fresh counter that counts down to the next visit of instruction $l_1$.
\end{enumerate}

It is easy to see that we can encode all the above as a \sohyperltl{} (or \fphyperltl{}) formula using only first-order quantification over traces in $X$. 
Let $\varphi$ be such a formula.
It is easy to see that the \sohyperltl{} formula $\exists X. \varphi$ holds on any system, iff there exists a set $X$ with the above properties iff the 2CM has a recurring computation. 
The reproves \Cref{prop:hyperltl_sat_in_sohyperltl_mc}.

For the present proof, we do, however, want to show $\Sigma^1_1$-hardness for the less powerful \fphyperltl{}.
The key point to extend this is to ensure that iff there exists a set $X$ that satisfies the above requirements, then there also exists a minimal one. 
The key observation is that -- by the construction of $X$ -- any set $X$ that satisfies the above \emph{is already minimal}:
The AP $\dagger$ ensures that, for each step, there exists exactly one configuration (\Cref{item:uni}), and, when removing any number of traces from $X$, we will inevitably violate \Cref{item:cons}.

We thus get that the 2CM has a recurring computation iff there exists \emph{mininmal} $X$ that satisfies $\varphi$ iff the \fphyperltl{} formula $\exists (X, \smallestSet, \varphi) \ldot \mathit{true}$ holds on an arbitrary system. \qed
\end{proof}

\sohyperToSat*
\begin{proof}
    Assume that we have $m$ second-order quantifiers, and for each $k \in [m]$, $\pi_1, \ldots, \pi_{l_k}$ are the first-order variables occurring before $X_k$ is quantified:
    \begin{align*}
        \varphi = \quant \pi_1, \ldots, \quant \pi_{l_1} \exists X_1 \quant \pi_{l_1 + 1}, \ldots, \quant \pi_{l_2} \exists X_2 \cdots \exists X_m \quant \pi_{l_m + 1}, \ldots, \quant \pi_{l_{m+1}} \psi\ , 
    \end{align*}
    The second-order variables we use thus are $\sovars = \{\systemvar, \allvar, X_1, \ldots, X_m\}$.
    Each $X \in \sovars$ can depend on some of the traces quantified before it. 
    In particular, each $X_k$ depends on traces $\pi_1, \ldots, \pi_{l_k}$, $\systemvar$ depends on none of the traces (as it is fixed) and neither does $\allvar$.
    We define a function $c : \sovars{} \to \nat$ that denotes on how many traces the set can depend, i.e., $c(\systemvar) = c(\allvar) = 0$, and $c(X_k) = l_k$.

    We then encode this functional dependence into a model by using traces. 
    For each $X \in \sovars$, we define atomic propositions
    \begin{align*}
        \ap_X := \{ [a, j ]_X \mid a\in \ap \land j \in [c(X)] \cup \{\dagger\} \}
    \end{align*}
    and then define 
    \begin{align*}
        \ap' := \ap \uplus \bigcup_{X \in \sovars{}} \ap_X
    \end{align*}

    We will use the original APs to describe traces. 
    The additional propositions are used to encode functions which map the $c(X)$ traces quantified before $X$ to some set of traces.
    Given traces $\pi_1, \ldots, \pi_{C(X)}$, we say a trace $t$ is in the model of $X$ if there exists some trace $\dot{t}$ (in the model of our final formula) such that 
    \begin{align}\label{eq:encoding}
    \begin{split}
        \forall z \in \nat\ldot \Big(\bigwedge_{j \in [C(X)]} \bigwedge_{a \in \ap} \big(a \in t_j(z) \leftrightarrow [a, j]_X \in \dot{t}(z) \big) \Big) \land \\
       \forall z \in \nat\ldot \Big(\bigwedge_{a \in \ap} \big(a \in t(z) \leftrightarrow [a, \dagger]_X \in \dot{t}(z) \big) \Big)
    \end{split}
    \end{align}
    holds.
    That is, the trace $\dot{t}$ defines the functional mapping from $t_1, \ldots, t_{c(X)}$ to $t$.
    
    We use this idea of encoding functions to translate $\varphi$ as follows:
    \begin{align*}
        \llbracket \psi \rrbracket &:= \psi\\
        \llbracket \quant X. \varphi \rrbracket &:= \llbracket \varphi \rrbracket\\
        \llbracket \exists \pi_i \in X\ldot \varphi \rrbracket &:= \exists \pi_i\ldot 
          \langle \pi_i \in X \rangle \land \llbracket \varphi \rrbracket\\
          \llbracket \forall \pi_i \in X\ldot \varphi \rrbracket &:= \forall \pi_i\ldot 
          \langle \pi_i \in X \rangle \rightarrow \llbracket \varphi \rrbracket
    \end{align*}
    We leave path formulas unchanged and fully ignore second-order quantification (which is always existential). 
    We define $\langle \pi \in X \rangle$ as an abbreviation for 
    \begin{align*}
        \exists \dot{\pi}\ldot \Big(\bigwedge_{j \in [c(X)]} \G \bigwedge_{a \in \ap} \big(a_{\pi_j} \leftrightarrow ([a, j]_X)_{\dot{\pi}}\big) \Big) \land \Big(\G \bigwedge_{a \in \ap} \big(a_{\pi} \leftrightarrow ([a, \dagger]_X)_{\dot{\pi}}\big) \Big)
    \end{align*}
    which encodes \Cref{eq:encoding}.

    The last thing we need to ensure is that $\systemvar$ and $\allvar$ are encoded correctly. 
    That is for any trace $t$ we have $t \in \traces(\calT)$ iff there exists a $\dot{t}$ (in the model) such that for any $z \in \nat$ and $a \in \ap$, we have $a \in t(z)$ iff $[a, \dagger]_\systemvar \in \dot{t}(z)$.
    Similarly, there should exists a trace $\dot{t}$ (in the model) such that for any $z \in \nat$ and $a \in \ap$, $[a, \dagger]_\allvar \in \dot{t}(z)$.
    Both requirement can be easily expressed as \hyperqptl{} formulas $\varphi_\systemvar$ and $\varphi_\allvar$ (Note that expressing these requirement in HyperLTL is not possible as we cannot quantify over traces that are not within the current model). 
    
    It is easy to see that $\calT \models \varphi$ iff $\llbracket \varphi \rrbracket \land \varphi_\systemvar \land \varphi_\allvar$ is satisfiable. \qed
\end{proof}

\easyfrags*

\begin{proof}
$\forall^* X \forall^*\pi.\varphi$, $\exists^* X \exists^*\pi.\varphi$:
As universal properties are downwards closed and existential properties are upwards closed, removing second order quantification does not change the semantics of the formula. 

$\exists^* X \forall^*\pi.\varphi$: for every variable $X$ we introduce a trace variable $\pi_x$ which is existentially bounded, and for every occurrence of $\pi$ in $\varphi$ such that $\pi\in X$, we replace $\pi$ with $\pi_x$. 
If $\varphi$ holds for all traces in $X$, it holds also when replacing $X$ with the singleton $\pi_x$. The other direction of implication is trivial as we found a set $X = \{ \pi_X \}$ for which $\varphi$ holds. 
For similar reasons, for $\forall^* X \exists^*\pi.\varphi$ we can remove the second order quantification and replace every existentially trace quantification with a universal trace quantification.

As a conclusion of all of the above, we have that the model-checking of \sohyperltl formulas of the following type is decidable: $\quant_1 X_1 \cdots \quant_k X_k. \quant'_1 \pi_1 \in X_1 \cdots \quant'_k \pi_k \in X_k. \varphi$ where  we have $\quant_i, \quant'_i\in \{\exists, \forall \}$ and in $\varphi$ only traces from the same set $X_i$ are compared to each other (that is, $\varphi$ does not bind traces from different sets to each other).  

Lastly, for $\psi = \exists X.\exists^* \pi\in X.\forall^*\pi'\in X. \varphi$, we use a reduction to the satisfiability problem of HyperQPTL~\cite{DBLP:conf/lics/CoenenFHH19}. Let $\varphi_{\aut{T}}$ a QPTL formula that models the system. Then, $\aut{T} \models \psi$ iff the HyperQPTL formula $\hat{\psi}$ is satisfiable, where $\hat{\psi} = \exists^* \pi. \forall^*\pi'.\forall\tau.\varphi_\aut{T}(\tau)\wedge\psi (\pi,\pi') $. Since $\hat{\psi}$ is a $\exists^*\forall^*$ HyperQPTL formula, its satisfiability problem is decidable~\cite{DBLP:conf/lics/CoenenFHH19}. \qed
\end{proof}

\section{Appendix for~\Cref{sec:examples}}\label{app:kltl}

\subsection{Additional Material for \Cref{sec:ltlK}} \label{app:ltlk}

LTL$_{\calK, \sfC}$ is defined by the following grammar:
\begin{align*}
    \psi := a \mid \neg \psi \mid \X \psi \mid \psi_1 \U \psi_2 \mid \calK_A \psi \mid \mathsf{E} \psi \mid \sfC \psi
\end{align*}
where $A$ is a set of agents.
Given two traces $t, t'$ we write $t[0,i] =_{A_i} t'[0, i]$ if $t$ and $t'$ appear indistinguishable for agent $A_i$ in the first $i$ steps.
Given a set of traces $T$ and a trace $t$ we define 
\begin{align*}
	t, i &\models_T  a &\text{iff} \quad  &a \in t(i)\\
	t, i &\models_T  \neg \psi &\text{iff} \quad & t, i \not\models_T  \psi \\
	t, i &\models_T  \psi_1 \land \psi_2 &\text{iff} \quad  &t, i \models_T \psi_1 \text{ and }  t, t \models_T  \psi_2\\
	t, i &\models_T  \X  \psi &\text{iff} \quad & t, i+1 \models_T \psi \\
	t, i &\models_T  \psi_1 \U \psi_2 &\text{iff} \quad & \exists j \geq i \ldot t, j \models_T\psi_2 \text{ and } \forall i \leq k < j \ldot  t, k \models_T  \psi_1\\
    t, i &\models_T  \calK_{A_i} \psi &\text{iff} \quad & \forall t' \in T\ldot t[0, i] =_{A_i} t'[0, i] \to t', i \models \psi\\
    t, i &\models_T \mathsf{E} \psi & \text{iff} \quad & t, i \models_T \bigwedge_{A_i \in A} \calK_{A_i} \psi\\
    t, i &\models_T  \sfC \psi &\text{iff} \quad & t,i \models_T \mathsf{E}^\infty \psi
\end{align*}
The \emph{everyone knows} operator $\mathsf{E}$ states that every agents knows that $\psi$ holds.
The semantics of the common knowledge operator $\sfC$ is then the infinite chain, or transitive closure, of \emph{everyone knows that everyone knows that ...} $\psi$.

\kltl*
\begin{proof}
    Let $\{ A_1, \ldots, A_n\}$ be the set of agents. For the $j$th occurrence of
    a knowledge operator $\know\in\{ \calK, \sfC \}$
    we introduce a new trace variable $\tau_j$ and a second order variable $Y_j$.  
    In addition, we introduce a new atomic proposition $k$.
    We then replace the $j$th occurrence of $\know$ in $\varphi$ with $k_{\tau_j}$, 
    resulting in the HyperLTL formula $\psi^\tau$.
    Denote by $\psi_{j}$ the subformula of $\psi^\tau$ that directly follows~$k_{\tau_{j}}$. 
     
    We define by induction on the nested knowledge operators the corresponding \fphyperltl{} formula $\varphi_j$. 
    For the first (inner-most) operator, we define $\varphi_0$ to be the LTL formula nested under this operator. Now, assume we have defined $\varphi_{j-1}$, and let $\know\in \{\calK, \sfC \}$ be the $j$th inner-most knowledge operator. Then, $\varphi_{j}$ is defined as follows. 

    \begin{align*}
        \forall \tau_j\in\allvar. \Big(Y_j, \smallestSet, (\pi\in Y_j \wedge \forall\pi_1\in Y_j .\forall\pi_2\in \systemvar. \left(\neg k_{\tau_j}\LTLuntil(k_{\tau_j}\wedge \LTLnext \LTLglobally \neg k_{\tau_j})\right) \rightarrow
        \\ 
       \left( \mathit{equiv}^j_{\know}(\pi_1, \pi_2) \LTLuntil k_{\tau_j}  \rightarrow
       \pi_2\modin Y_j \right)\Big) . 
       \forall \pi_1\in Y_j.\varphi_{j-1} \wedge \LTLglobally \left(k_{\tau_j} \rightarrow \psi_j \right)
    \end{align*}
    Where
    \begin{align*}
        \mathit{equiv}^j_{\calK_{A_i}} :=  \pi_1 \leftrightarrow_{A_i} \pi_{2}  
          \quad\quad\quad \mathit{equiv}^j_{\sfC}: =   \bigvee_{i\in[n]} \pi_1 \leftrightarrow_{A_i} \pi_{2}   
    \end{align*}

    Each instantiation of the universally quantified variable $\tau_j$ corresponds to one timepoint in which we want to check knowledge on. Therefore, we verify that $k$ appears exactly once on the trace (first line of the formula). Then, we add to the knowledge set all traces that are equivalent (by the knowledge if this agent, or by the common knowledge of all agents) until this timepoint. 
    The formula outside the minimality condition verifies that for all traces in the set $Y_j$, the subformula $\varphi_{j-1}$ holds thus enforcing the knowledge requirement on all traces in $Y_j$. In addition, it uses the property $\LTLglobally \left(k_{\tau_j} \rightarrow \psi_j \right)$ to make sure that the temporal (non-knowledge) requirements hold at the same time for all traces in $Y_j$. 
 Finally, we define $\varphi' = \forall\pi.\varphi_n$ where $n$ is the number of knowledge operators in~$\varphi$. Note that in general, the formula above can yield infinitely many sets of traces. In practical examples, e.g. the examples appear in this paper, we can write simplified formulas that reason about the specific problem at hand and only require a finite (usually $1$) number of such sets. Also note that as the sets $Y_j$ are unique, we do not need the quantification over least sets. 
    \qed
\end{proof}

\subsection{Additional Material To \Cref{sec:asynchronous:hyperproperties}}\label{app:asynchronous:hyperproperties}

Consider the system in \Cref{fig:syn_program} (taken from \cite{BaumeisterCBFS21}).
The synchronous version of observational determinism ($\varphi_\mathit{OD}$) holds on this system:
While we branch on the secret input $h$, the value of $o$ is the same across all traces.
In contrast,  $\varphi_\mathit{OD}$ does not hold on the system in \Cref{fig:asyn_program} as, in the second branch, the update occurs one step later. 
This, however, is not an accurate interpretation of $\varphi_\mathit{OD}$ (assuming that an attacker only has access to the memory footprint and not the CPU registers or a timing channel), as any two traces are \emph{stutter} equivalent (with respect to $o$).
In AHLTL we can express an asynchronous version of OD as $\forall \pi_1. \forall \pi_2. \mathbf{E} \G(o_{\pi_1} \leftrightarrow o_{\pi_2})$ stating that all two traces can be aligned such they (globally) agree on $o$.
This formula now holds on both \Cref{fig:syn_program} and \Cref{fig:asyn_program}.

\begin{figure}[!t]
\vspace{-5mm}
\begin{subfigure}{0.5\textwidth}
\begin{exampleCode}
$l$ <- $0$
if $h$ then 
@@$o$ <- 1
else 
@@$o$ <- $o$ + 1
\end{exampleCode}
\subcaption{}\label{fig:syn_program}
\end{subfigure}%
\begin{subfigure}{0.5\textwidth}
\begin{exampleCode}
$o$ <- $0$
if $h$ then 
@@$o$ <- 1
else 
@@$reg$ <- $o$ + 1
@@$o$ <- $reg$
\end{exampleCode}
\subcaption{}\label{fig:asyn_program}
\end{subfigure}

\caption{Example Programs.}
\vspace{-5mm}
\end{figure}

\section{Additional Material for \Cref{sec:implementation}}\label{app:implementation}

\subsection{Muddy Children}

We consider the following \fphyperltl{} formula which captures the common knowledge set after $m$ steps. 
\begin{align*}
    &\forall \pi \in \systemvar{}. \\
    &\quad\Big(X, \smallestSet,\pi \modin X \land \forall \pi_1 \in X. \forall \pi_2 \in \systemvar{}. \big(\bigvee_{i \in [n]} \G^{\leq m} \pi_1 =_{\ap_i} \pi_2\big) \Rightarrow \pi_2 \modin X \Big). \\
    &\quad\quad\quad\quad\forall \pi_1 \in X. \forall \pi_2 \in X. \bigwedge_{i \in [n]} {i_c}_{\pi_1} \leftrightarrow {i_c}_{\pi_2}
\end{align*}
where we write $\G^{\leq m} \psi$ to assert that $\psi$ holds in the first $m$ steps. 
Here, $\ap_i$ are all propositions observable by child $i$, i.e., all variables expect the one that determines if $i$ is muddy.
Note that this formula falls within our fixpoint fragment of \fphyperltl{}.

Further note that we express a hyperproperty on the knowledge set, i.e., compare pairs of traces in the knowledge set. 
This is not possible in logic's such as LTL$_{\calK, \sfC}$ in which we can only check if a trace property holds on the knowledge set.

\subsection{Asynchronous Hyperproperties}

We verify 
\begin{align}\label{eq:od_synch}
    \varphi_\mathit{OD} := \forall \pi_1. \forall \pi_2. \G (o_{\pi_1} \leftrightarrow o_{\pi_2})
\end{align}
and the asynchronous version of it. 
In AHLTL \cite{BaumeisterCBFS21} we can define this as follows:
\begin{align*}
    \varphi_\mathit{OD} := \forall \pi_1. \forall \pi_2. \mathbf{E} \G (o_{\pi_1} \leftrightarrow o_{\pi_2})
\end{align*}
In \fphyperltl{} we can express the above AHLTL formula as the following formula:
\begin{align}\label{eq:od_asynch}
\begin{split}
    &\forall \pi_1 \in \systemvar. \forall \pi_2 \in \systemvar. \\
    &\quad\big(X_1, \smallestSet, \pi_1 \modin X_1 \land \forall \pi \in X. \forall \pi \in \allvar. ((o_{\pi} \leftrightarrow o_{\pi'}) \U \G (o_{\pi} \leftrightarrow \X o_{\pi'}) \rightarrow \pi' \modin X_1 \big)\\
    &\quad\big(X_2, \smallestSet, \pi_2 \modin X_2 \land \forall \pi \in X. \forall \pi' \in \allvar. ((o_{\pi} \leftrightarrow o_{\pi'}) \U \G (o_{\pi} \leftrightarrow \X o_{\pi'}) \rightarrow \pi' \modin X_2 \big)\\
    &\quad\quad\exists \pi_1 \in X_1, \exists \pi_2 \in X_2\ldot \G (o_{\pi_1} \leftrightarrow o_{\pi_2})
\end{split}
\end{align}
Note that this formula falls within our fixpoint fragment of \fphyperltl{}.

In \Cref{tab:more}, we check \Cref{eq:od_synch} and \Cref{eq:od_asynch} on the two example programs from the introduction in \cite{BaumeisterCBFS21} and the asynchronous program \texttt{Q1} from \cite{BeutnerF21,BeutnerF23LMCS}.

\subsection{Mazurkiewicz Trace}

Using our logic, we can also express many properties that reason about the class of (Mazurkiewicz) traces. 
The idea of  trace is to abstract away from the concrete order of independent actions (letters). 
Let $I \subseteq \Sigma \times \Sigma$ be an independence relation on letters. That is, $(a, b) \in I$ iff interchanging the order of $a$ and $b$ has no effect (e.g., local actions for two concurrent processes). 
We say two traces $t_1, t_2$ are equivalent (written $t_1 \equiv_I t_2)$ if we can rewrite $t_1$ into $t_2$ by flipping consecutive letters that are in $I$.
For example if $(a, b) \in I$ then $xaby \equiv_I xbay$ for all $x \in \Sigma^*, y \in \Sigma^\omega$.
The Mazurkiewicz trace of a concrete trace $t$ is then defined as $[t]_I := \{t' \mid t \equiv_i t\}$. 

Using \fphyperltl{} we can directly reason about the equivalence classes of~$\equiv_I$.
Consider the following (quantifier-free) formula $\varphi_I(\pi, \pi')$, stating that $\pi$ and $\pi'$ are identical up to one flip of consecutive independent actions. 
\begin{align*}
   (\pi =_\ap \pi') \W \Big(\bigvee_{(x, y) \in I} x_\pi \land y_{\pi'} \land \X (y_\pi \land x_{\pi'}) \land \X\X \G (\pi =_\ap \pi') \Big)
\end{align*}
Here we write $x_\pi$ for $x \in \Sigma = 2^\ap$ for the formula $\bigwedge_{a \in X} a_\pi \land \bigwedge_{a \not\in X} \neg a_\pi$.
The formula asserts that both traces are equal, until one pair of independent actions is flipped, followed by, again, identical suffixes. 

Using $\varphi_I$ we can directly reason about Mazurkiewicz traces.
Assume we are interested if for every (concrete) trace $t$ that satisfies LTL property $\phi$, all its equivalent traces satisfy $\psi$.
We can express this in \fphyperltl{} as follows:
\begin{align}\label{eq:trace}
\begin{split}
     &\forall \pi \in \systemvar{}. (X, \smallestSet, \pi \modin X \land \forall {\pi_1} \in X\ldot \forall {\pi_2} \in \allvar\ldot \varphi_I(\pi_1, \pi_2) \rightarrow \pi_2 \modin X) \ldot \\
    &\quad\quad\forall \pi' \in X\ldot \phi(\pi) \rightarrow \psi(\pi')
\end{split}
\end{align}
That is, for all traces $\pi$, we compute the set $X$ which contains all equivalent traces. 
This set should contain $\pi$, must be closed under $\varphi_I$, and is minimal w.r.t.~those two properties. 
Note that this formula falls within our fixpoint fragment of \fphyperltl{}.

\subsubsection*{Preliminary Experiments.}
The above formula is applicable to arbitrary independec relation, so our tool can be used to automatically check arbitrary properties on Mazurkiewicz traces.
In our preliminary experiments, we focus on very simple Mazurkiewicz trace.
We model a situation where we start with $\{a\}^1\{\}^\omega$ and can swap letters $\{a\}$ and $\{\}$.
We acknowledge that the example is very simple, but nevertheless emphasize the complex trace-based reasoning possible with \texttt{HySO}.
We then ask the following problems: 
We ask if, from every trace, $a$ holds at most once on all traces in $X$. 
If we apply only iteration, \texttt{HySO} will move the unique $a$ one step to the right in each iteration, i.e., after $n$ steps the current under-approximation contains traces $\{a\}\{\}^\omega$, $\{\}^1\{a\}\{\}^\omega$, $\{\}^2\{a\}\{\}^\omega$, $\ldots$, $\{\}^n\{a\}\{\}^\omega$. 
This fixpoint will never converge, so \texttt{HySO} would iterate forever. 
Instead, if we also enable learning of overapproximations,  \texttt{HySO} automatically learns an invariant that proves the above property (which is instance \textsc{SwapA} in \Cref{tab:mar}). 
We can relax this requirement to only consider a fixed, finite prefix. 
\textsc{SwapA}$_n$ states that $a$ \emph{can} hold in any position in the first $n$ steps (by using existential quantification over $X$). 
As expected, \texttt{HySO}  can prove this property by iteration $n$ times.
Lastly, the violation \textsc{SwapAViolation}$_n$ states that any trace satisfies $a$ within the first $n$ steps. 
This obviously does not hold, but \texttt{HySO} requires $n$ iterations to find a counterexample, i.e., a trace where $a$ does not hold within the first $n$ steps.

\end{document}